\documentclass[twocolumn,aps,amsmath,amssymb,floatfix,superscriptaddress,prx,nofootinbib,longbibliography]{revtex4-2}
\usepackage{amsmath,amssymb,amstext}
\usepackage[ruled,lined]{algorithm2e}
\usepackage{algpseudocode}
\usepackage{amsthm}
\usepackage{bm}
\usepackage{bbm}
\usepackage{blochsphere}
\usepackage{braket,physics,physics2}
\usepackage{comment}
\usepackage{enumerate}
\usepackage{framed}
\usepackage{graphicx}
\usepackage[colorlinks,linkcolor=blue,anchorcolor=blue,citecolor=blue,urlcolor=blue]{hyperref}
\usepackage{ifthen}
\usepackage[utf8]{inputenc}
\usepackage{listliketab}
\usepackage{lmodern}
\usepackage{mathrsfs}
\usepackage{mathtools}
\usepackage{natbib}
\usepackage{setspace}
\usepackage{standalone}
\usepackage[caption=false,labelformat=simple]{subfig}
\usepackage{svg}
\usepackage{tikz}
\usepackage{tikz-3dplot}
\usepackage{xcolor}
\usepackage{multirow}
\usepackage{accents}
\usephysicsmodule{ab,ab.braket}

\everymath{\displaystyle}
\newtheorem{definition}{Definition}
\newtheorem{theorem}{Theorem}
\newtheorem{lemma}{Lemma}

\newcommand{\alg}[1]{\hyperref[alg:#1]{Algorithm~\ref*{alg:#1}}}
\newcommand{\defn}[1]{\hyperref[defn:#1]{Definition~\ref*{defn:#1}}}
\newcommand{\thm}[1]{\hyperref[thm:#1]{Theorem~\ref*{thm:#1}}}
\newcommand{\lem}[1]{\hyperref[lem:#1]{Lemma~\ref*{lem:#1}}}
\newcommand{\prop}[1]{\hyperref[prop:#1]{Proposition~\ref*{prop:#1}}}
\newcommand{\tab}[1]{\hyperref[tab:#1]{TABLE~\ref*{tab:#1}}}
\newcommand{\fig}[1]{\hyperref[fig:#1]{Fig.~\ref*{fig:#1}}}
\newcommand{\figg}[2]{\hyperref[fig:#1]{Fig.~\ref*{fig:#1}#2}}
\renewcommand{\sec}[1]{\hyperref[sec:#1]{Sec.~\ref*{sec:#1}}}
\newcommand{\app}[1]{\hyperref[app:#1]{App.~\ref*{app:#1}}}

\DeclareMathOperator{\polylog}{polylog}

\usetikzlibrary{quantikz2}
\tikzset{operator/.append style={rounded corners}}

\newboolean{ShowComments}
\setboolean{ShowComments}{true} 
\ifthenelse{\boolean{ShowComments}}%
	{
		\newcommand{\ColorComment}[3]{%
				{\colorbox{#1}{\color{white}   \textsf{\textbf{#2}}} \textcolor{#1}{#3}}}

	}%
	{
		\newcommand{\ColorComment}[3]{}

	}%

\definecolor{miyashitacolor}{cmyk}{1,0,0,0}

\makeatletter
\newif\ifappendixtoc
\appendixtocfalse

\let\orig@addcontentsline\addcontentsline
\renewcommand{\addcontentsline}[3]{%
  \def\atoc@file{#1}%
  \def\atoc@toc{toc}%
  \ifappendixtoc
    \ifx\atoc@file\atoc@toc
      \orig@addcontentsline{atoc}{#2}{#3}%
    \else
      \orig@addcontentsline{#1}{#2}{#3}%
    \fi
  \else
    \orig@addcontentsline{#1}{#2}{#3}%
  \fi
}

\makeatletter
\newcommand{\appendixtableofcontents}{%
  \begingroup
    \let\saved@addcontentsline\addcontentsline
    \let\addcontentsline\@gobblethree
    \section*{Appendices}%
    \let\addcontentsline\saved@addcontentsline

    \@starttoc{atoc}%
  \endgroup
}
\makeatother

\begin{document}

\title{Quantum-accelerated conjugate gradient method via spectral initialization}

\author{Shigetora Miyashita}
\email{shigetora.miyashita@g.softbank.co.jp}
\affiliation{Research Institute of Advanced Technology, SoftBank Corp., Tokyo Portcity Takeshiba Office Tower 1-7-1, Kaigan, Minato-ku, Tokyo, 105-7529, Japan}
\author{Yoshi-aki Shimada}
\email{yoshiaki.shimada01@g.softbank.co.jp}
\affiliation{Research Institute of Advanced Technology, SoftBank Corp., Tokyo Portcity Takeshiba Office Tower 1-7-1, Kaigan, Minato-ku, Tokyo, 105-7529, Japan}

\date{\today}

\begin{abstract}
Solving large-scale linear systems problems is a cornerstone in scientific and industrial computing.
Classical iterative solvers face increasing difficulty as the number of unknowns becomes large, while fully quantum linear solvers require fault-tolerant resources that remain far beyond near-term feasibility.
Here we propose a quantum-accelerated conjugate gradient (QACG) method in which a fault-tolerant quantum algorithm is used exclusively to construct a spectrally informed initial guess for a classical conjugate gradient (CG) solver.
We estimate the total runtime and resource requirements of an integrated quantum-HPC platform for the 3D Poisson equation.
A central feature of QACG is the controllable decomposition of the condition number between the quantum and the classical solver, enabling flexible allocation of computational effort.
Under explicit architectural assumptions, we identify regimes in which QACG yields a runtime advantage over purely classical approaches while requiring substantially fewer quantum resources than end-to-end quantum linear solvers.
These results illustrate a concrete pathway toward the scientific and industrial use of early-stage fault-tolerant quantum computing and point to a integrated paradigm in which quantum devices act as accelerators within high-performance computing workflows.
\end{abstract}

\maketitle

\section{Introduction}
\label{sec:introduction}
Quantum algorithms present a unique path for computation by introducing a set of elementary building blocks, based on unitary operations acting on quantum states, that can be used to accelerate specific computational tasks. From this perspective, quantum computation is not intended to replace classical algorithms. Rather, it enables efficient solutions to specific problems that are prohibitively costly in the classical setting, by exploiting physical phenomena such as superposition and entanglement. This perspective traces back to Feynman’s insight that quantum dynamics cannot be efficiently simulated using classical resources alone~\cite{Feynman1982}, demonstrating that quantum mechanics enable information processing that lie beyond the reach of classical computation. Subsequent developments, most notably Shor’s factoring algorithm~\cite{shor1999polynomial,ekert1996quantum,monz2016realization} and Grover’s search algorithm~\cite{grover1996fast,mandviwalla2018implementing}, showed that carefully designed quantum subroutines can be embedded within larger computational workflows to achieve substantial performance gains. More recently, advances in quantum error correction~\cite{lidar2013quantum,Devitt_2013,Roffe_2019,2024} have elevated this paradigm into a rigorous theoretical framework, enabling a systematic assessment of quantum computation. These developments have reshaped the boundary between classical and quantum computation and delineated the regimes in which quantum advantage may arise.

Recent attention has shifted toward noisy intermediate-scale quantum (NISQ) devices~\cite{Preskill2018quantumcomputingin}, which operate without full error correction. In this regime, collaborative quantum--classical algorithms, such as quantum-selected configuration interaction (QSCI)~\cite{kanno2023quantum} and sample-based quantum diagonalization (SQD)~\cite{Robledo_Moreno_2025,yu2025quantum}, have enabled practical applications including the estimation of molecular ground-state energies. When embedded in workflows supported by high-performance computing (HPC) resources, these methods have extended simulations beyond the memory limits of purely classical architectures. In parallel, fault-tolerant quantum computing (FTQC)~\cite{gottesman1998theory,shor1996fault,preskill1998fault,aharonov1997fault} offers provably correct solutions to selected quantum chemistry problems. For example, few-qubit demonstrations of quantum phase estimation (QPE)~\cite{yamamoto2025quantumerrorcorrectedcomputationmolecular} have produced certified ground-state energies. Because QPE also appears as a subroutine in a broad range of algorithms including the Harrow--Hassidim--Lloyd (HHL) algorithm~\cite{harrow2009quantum}, its relevance extends across physics, chemistry, and engineering.

Despite this progress, both paradigms face substantial limitations. In NISQ-based quantum-centric workflows, it remains unclear if quantum resources can provide a guaranteed advantage in general tasks. Problem-specific techniques, such as the local unitary cluster Jastrow (LUCJ) ansatz~\cite{motta2023bridging}, achieve impressive results in electronic structure calculations but do not readily generalize. Conversely, many FTQC-oriented proposals restrict the role of classical HPC to auxiliary preprocessing or postprocessing tasks~\cite{doi:10.1126/science.adt0019}. Although error correction is indispensable for large-scale quantum computation, existing demonstrations have not yet addressed industrial-scale problems end to end. Taken together, these observations suggest that neither classical HPC nor near-term or fault-tolerant quantum computing alone can efficiently address the full spectrum of large-scale scientific and industrial workloads. Classical HPC struggles with problems dominated by severe ill-conditioning or extreme spectral complexity, while quantum algorithms applied monolithically often incur prohibitive resource requirements.

These considerations motivate a central question: does there exist a computational regime beyond monolithic architectures, whether purely classical or purely quantum, that can deliver measurable and practically relevant benefits? In this work, we use the term \emph{beyond monolithic architecture} to denote a setting in which classical and quantum resources are combined within an integrated quantum--HPC system, each addressing the aspects of a problem for which it is best suited, without claiming universal or unconditional advantage.

The objective of this study is to quantitatively identify the conditions under which such a regime can be realized for large-scale linear systems. We propose an integrated fault-tolerant quantum--HPC algorithm in which the HHL algorithm is used to construct a spectral initial guess for the classical conjugate gradient (CG) method. The quantum subroutine is restricted to a low-energy spectral subspace that is responsible for slow classical convergence, while the classical HPC solver performs the bulk of the numerical workload. In this workflow, the quantum component mitigates the most ill-conditioned modes, and the classical CG iteration exploits its scalability in the remaining subspace. As a result, the effective condition number governing the quantum stage is reduced, and the total number of classical floating-point operations is correspondingly lowered for problems with a large number of unknowns~$N$.

We analyze the performance of this integrated quantum--HPC system using explicit gate-complexity and runtime models. The quantum component is evaluated within a partially fault-tolerant framework based on the STAR architecture~\cite{PRXQuantum.5.010337}, while the classical component is modeled on a contemporary HPC platform. Importantly, our analysis does not assume idealized end-to-end quantum acceleration; rather, it treats the method as a cooperative workflow whose advantage is conditional on problem structure, spectral properties, and architectural parameters.

This perspective aligns with increasing interest in quantum--HPC integration, where quantum processors are incorporated into existing HPC environments as specialized accelerators~\cite{rallis-2025-interfacing,mansfield-2025-practical-experiences}. While such integration enables realistic assessments of hybrid workflows, the classes of applications that can benefit, and the nature of those benefits, remain open questions.

Focusing on computer-aided engineering (CAE)~\cite{raphael2003fundamentals}, an industrially important domain dominated by large sparse linear systems, we present a quantum-accelerated conjugate gradient (QACG) method as a concrete case study. By demonstrating a conditional performance advantage for a representative CAE problem, this work illustrates how cooperation between quantum and classical resources can expand the range of tractable computations beyond what is achievable with either paradigm in isolation.

The remainder of this paper is organized as follows. In \sec{solving_linear_systems_of_equations}, we review the linear systems problem and the relevant classical and quantum algorithms, emphasizing their computational complexity. In \sec{quantum_accelerated_conjugate_gradient_method}, we introduce the QACG workflow and present an explicit quantum circuit for a three-dimensional Poisson equation. In \sec{performance_estimation_for_quantum-HPC_platforms}, we report runtime and resource estimates in terms of logical qubit counts and gate complexity as well as the numerical simulation to demonstrate the effectiveness of our approach. We conclude in \sec{conclusion} with implications, limitations, and open challenges.

Throughout this work, $\norm{\cdot}$ denotes the matrix and vector $2$-norm, that is, $\norm{\cdot}_2$. We use $\log(\cdot)$ to denote the base-$2$ logarithm, $\log_2(\cdot)$, and $\ln(\cdot)$ to denote the natural logarithm.

\section{Solving linear systems of equations}
\label{sec:solving_linear_systems_of_equations}
A primary focus of this work is the solution of large-scale systems of linear equations of the form
\begin{align}\label{eq:LSP}
    A x = b,
\end{align}
where $A$ is a sparse Hermitian matrix and $b$ is a given right-hand-side vector. Such linear systems arise ubiquitously in scientific computing, most notably through the discretization of boundary value problems governed by partial differential equations (PDEs). Applications ranging from fluid dynamics and electromagnetism to structural analysis and materials science rely on PDE models to describe spatially varying physical quantities. Discretization by finite difference~\cite{doi:10.1137/1.9780898717839}, finite volume~\cite{eymard2000finite}, or finite element~\cite{jagota2013finite} methods yields algebraic systems whose dimension reflects the resolution of the underlying mesh. As this resolution increases, the number of unknowns $N$ grows rapidly, leading to large and computationally demanding linear systems that often dominate the overall cost of CAE simulations on modern HPC platforms.

\subsection{Method of conjugate gradients}
For linear systems of the form~\eqref{eq:LSP}, Krylov subspace methods are the standard solvers in large-scale numerical linear algebra. Among them, CG is widely used when $A$ is symmetric positive definite (SPD). CG can be derived from the Lanczos process by applying the Ritz--Galerkin condition to the quadratic functional
\begin{align}
    f(x)=\frac12 x^T A x - b^T x,
\end{align}
whose minimization is equivalent to solving $Ax=b$. Assuming sparse matrix storage and efficient memory access, each CG iteration requires $\mathcal{O}(N)$ operations. The number of iterations required to reach a target accuracy $\varepsilon$ depends on the condition number $\kappa$ of $A$, leading to the overall complexity
\begin{align}\label{eq:complexity-cg}
    \mathcal{T}_\text{CG}
    = \mathcal{O}\left(N\sqrt{\kappa}\log \varepsilon^{-1}\right).
\end{align}

The convergence behavior of CG follows from classical Krylov subspace theory. After $k$ iterations, the error satisfies
\begin{align}
    \|e^{(k)}\|_A
    = \min_{p\in\mathcal{P}_k,p(0)=1}\|p(A)e^{(0)}\|_A,
\end{align}
where $\mathcal{P}_k$ denotes the set of polynomials of degree at most $k$ and
$\|v\|_A\coloneqq\sqrt{v^TA v}$ is the energy norm. Choosing $p$ from the Chebyshev family yields the classical bound~\cite{10.5555/865018}
\begin{align}\label{eq:standard-CG-bound}
    \frac{\|e^{(k)}\|_A}{\|e^{(0)}\|_A}
    \le 2\left(\frac{\sqrt{\kappa}-1}{\sqrt{\kappa}+1}\right)^k,
\end{align}
which shows that CG convergence is governed by $\sqrt{\kappa}$. This dependence constitutes a fundamental bottleneck: even on large-scale HPC systems, the presence of small eigenvalues can significantly slow convergence.

A standard strategy to mitigate this bottleneck is preconditioning, which aims to reduce the effective condition number. Multiscale and multigrid methods are prominent examples. By exploiting a hierarchy of discretizations, multigrid techniques eliminate low-frequency error components on coarse grids and high-frequency components on fine grids. A generic multigrid update can be written as~\cite{doi:10.1137/S1064827597327310}
\begin{align}
    x^{(k+1)} = x^{(k)} + R_{(k)}\bigl(b-A_{(k)}x^{(k)}\bigr),
\end{align}
where $R_{(k)}$ approximates the inverse of $A_{(k)}$ at level $k$. Under ideal assumptions, including ellipticity, smooth coefficients, and well-defined grid hierarchies, multigrid methods can achieve an effective condition number $\kappa=\mathcal{O}(1)$ and overall complexity $\mathcal{O}(N)$. In practice, however, this performance is highly problem dependent. Irregular geometries~\cite{DARBANDI2006321}, discontinuous coefficients~\cite{NASTASE2006330}, or non-separable operators~\cite{10.1007/BFb0069947} may prevent the construction of effective hierarchies~\cite{Bank1988}, leading to substantial degradation in performance.

Adaptive mesh refinement (AMR) provides a complementary strategy by increasing resolution only where required. Refinement is guided by a posteriori error indicators, such as~\cite{LI20103139}
\begin{align}
    \frac{\left| x_{i+1}-2x_i+x_{i-1} \right|}
    {\left| x_{i+1}-x_i \right|+\left| x_i-x_{i-1} \right|
    +\varepsilon\left(\left| x_{i+1} \right|+2\left| x_i \right|+\left| x_{i-1} \right|\right)},
\end{align}
with elements exceeding a prescribed threshold marked for refinement. Although AMR can substantially reduce the number of degrees of freedom required to achieve a target accuracy, it introduces algorithmic and implementation complexity~\cite{hannoun2007issues}. Dynamic data structures, irregular communication patterns, and load imbalance can limit scalability on HPC systems~\cite{fluids10050129}. Moreover, AMR relies on assumptions about solution regularity that may fail in turbulent, strongly nonlinear, or tightly coupled multiphysics problems.

\subsection{Quantum linear systems algorithm}

The family of quantum linear systems algorithms (QLSA)~\cite{dervovic2018quantum,morales2025quantumlinearsolverssurvey,Lee_2019,chia2018quantuminspiredsublinearclassicalalgorithms,2020PhLA..38426595D} aims to prepare a quantum state proportional to the solution of the linear system~\eqref{eq:LSP},
\begin{align}
    \ket|x> = \frac{A^{-1}\ket|b>}{\norm{A^{-1}\ket|b>}},
\end{align}
rather than explicitly outputting the solution vector. This formulation illustrates a key distinction from classical solvers such as CG. While CG produces a solution with a runtime that scales linearly in the problem size $N$ as presented by~\eqref{eq:complexity-cg}, QLSA target a quantum state representation whose preparation cost depends only logarithmically on $N$. In this sense, QLSA offer an exponential improvement in the dependence on system size $N$.

This advantage, however, must be weighed against the dependence on the condition number $\kappa$. Classical CG converges in a number of iterations proportional to $\sqrt{\kappa}$, making it highly effective for well-conditioned or moderately conditioned problems. By contrast, quantum algorithms face a dependence on $\kappa$ and the target precision $\varepsilon$, shifting the computational bottleneck from system size to spectral properties.

Within the family of QLSA, the original HHL algorithm~\cite{harrow2009quantum} and later refinements achieve different trade-offs between these parameters. In terms of gate or query complexity, HHL scales as
\begin{align}
    \mathcal{G}_\text{HHL}
        = \mathcal{O}\left(\log N\kappa^{2}\varepsilon^{-1}\right),
\end{align}
while more modern QLSA formulations based on block-encoding and quantum singular value transformation achieve the improved scaling~\cite{PRXQuantum.3.040303}
\begin{align}
    \mathcal{Q}_\text{QLSA}
        = \mathcal{O}\left(\log N\kappa\log\varepsilon^{-1}\right).
\end{align}
The latter represents a significant theoretical improvement, reducing the quadratic dependence on $\kappa$ to linear and improving the precision dependence from polynomial to logarithmic.

Despite this asymptotic advantage, the dependence on $\kappa$ remains a central challenge for quantum solvers, particularly for the ill-conditioned systems commonly encountered in CAE. In such settings, the condition number often grows rapidly with problem size, offsetting the exponential improvement in $N$ and limiting the practical applicability of fully quantum approaches.

As in classical numerical linear algebra, preconditioning therefore plays a crucial role. Classical techniques such as sparse approximate inverse (SPAI) and wavelet-based preconditioners seek to transform~\eqref{eq:LSP} into
\begin{align}
    M^{-1}Ax = M^{-1}b,
\end{align}
where $M \approx A^{-1}$ clusters the spectrum of $M^{-1}A$ and suppresses low-energy modes that dominate convergence~\cite{PhysRevLett.110.250504}. From a spectral viewpoint, these methods selectively mitigate the components responsible for large $\kappa$. This perspective closely parallels the motivation for the quantum subroutines considered in this work, which likewise aim to address the most problematic spectral components.

However, when implemented on fault-tolerant quantum hardware, preconditioning introduces additional circuit and space overheads. Modern QLSA achieve their favorable complexity by assuming efficient access to block-encodings or linear combination of unitaries (LCU)~\cite{Childs2012LCU} representations of $A$, implemented via SELECT and PREPARE oracles. While powerful, these constructions require ancillary registers whose size grows with the complexity of encoding the matrix coefficients. As demonstrated by explicit circuit~\cite{camps2023explicitquantumcircuitsblock}, the resulting ancilla overhead and control logic can dominate the overall space complexity, even when the oracle query complexity is asymptotically optimal.

In contrast, the original HHL algorithm operates directly through Hamiltonian simulation of $A$ and does not require an explicit block-encoding or LCU representation. Although this leads to a less favorable asymptotic dependence on $\kappa$, it avoids the additional ancilla qubits needed to encode matrix coefficients. This reduced space overhead makes HHL particularly well suited for exploration in early fault-tolerant regimes, where logical qubits are scarce and architectural constraints are severe.

For these reasons, we focus on HHL rather than fully block-encoding-based QLSA in this work. Our goal is not to compete with the asymptotically optimal QLSA framework, but to assess the practical utility of quantum spectral techniques under realistic device constraints, and to understand how quantum and classical resources can be combined effectively within an integrated quantum--HPC platform.
\section{Quantum accelerated conjugate gradient method}
\label{sec:quantum_accelerated_conjugate_gradient_method}
In this section, we provide an integrated fault-tolerant quantum--HPC algorithm for LSP to address large-scale linear systems more efficiently than either resource alone under explicit spectral conditions. The strategy is based on the spectral decomposition of the coefficient matrix~$A$. A quantum computer is employed to approximate the solution in a low-energy subspace via HHL, while CG is employed in the complementary subspace. We refer to this combined strategy as QACG. QACG is designed as a cooperative integrated quantum--HPC system: the quantum subroutine performs a spectrally targeted initialization, while the classical CG solver carries out the main large-scale iterative refinement. We follow the notation used in Ref.~\cite{10.5555/865018}.

\begin{figure*}[htbp]
    \centering
    \includegraphics[height=8cm]{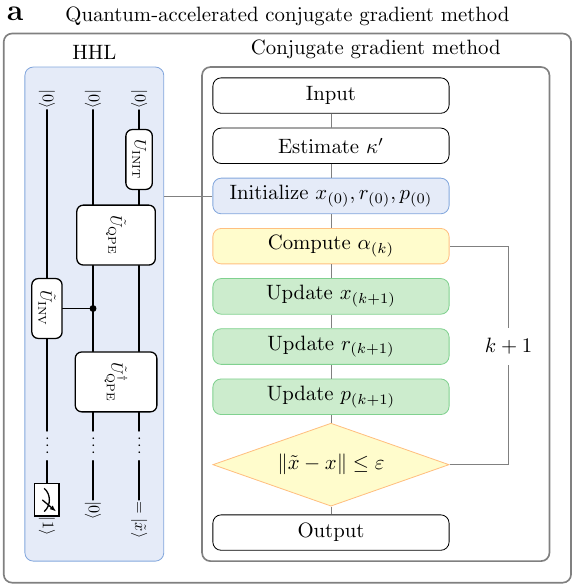}
    \includegraphics[height=8cm]{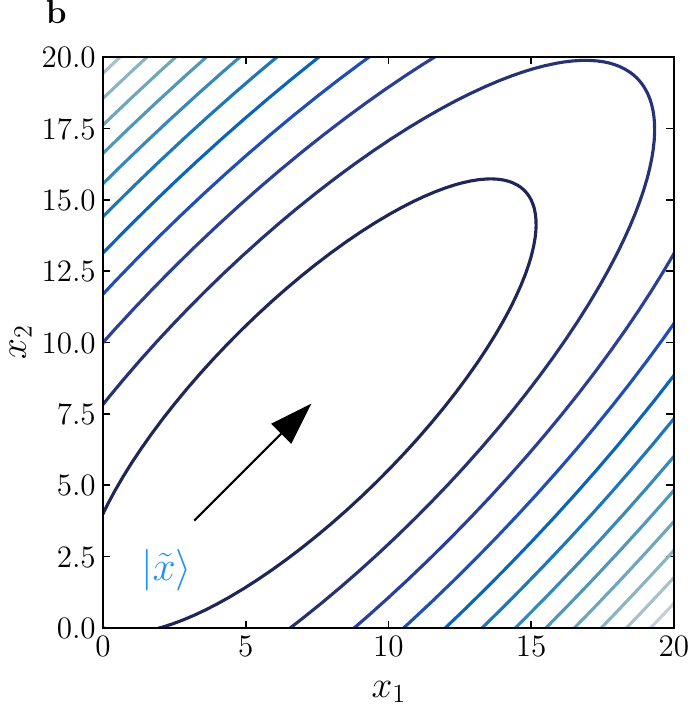}
    \caption{Quantum-accelerated conjugate gradient (QACG) method. \textbf{a}, The algorithm proceeds in two stages. In the first stage, a quantum routine based on the HHL algorithm prepares an initial approximation $\ket|\tilde{x}>$ by encoding the solution in a low-energy subspace. In the second stage, a classical conjugate gradient (CG) method refines this initial guess: starting from $x^{(0)}$, it iteratively computes the step size $\alpha^{(k)}$, updates the solution $x^{(k+1)}$, residual $r^{(k+1)}$, and search direction $p^{(k+1)}$, and terminates when $\norm{\tilde{x}-x}\le\varepsilon$. \textbf{b}, The schematic illustrates the CG search on a quadratic surface, where contour lines represent the objective function and the arrow indicates the descent trajectory. The initialized quantum states $\ket|\tilde{x}>$ place the starting point closer to the minimum, thereby shortening the path of the classical iterations and accelerating convergence.}
    \label{fig:qacg}
\end{figure*}

\subsection{Warm-start conjugate gradients}

When $A$ is SPD, CG can be interpreted as minimizing a quadratic function induced by the $A$-norm. Specifically, at iteration $k$, the approximation $x_{(k)}$ minimizes
\begin{align}
    \norm{x-x_{(k)}}_A^2 \coloneqq \left( x - x_{(k)} \right)^T A \left( x - x_{(k)} \right)
\end{align}
over the affine subspace determined by the previous iterates. Equivalently, CG minimizes the quadratic functional
\begin{align}
    f(x) \coloneqq \frac{1}{2}x^T A x - b^T x,
\end{align}
whose first-order optimality condition is $Ax=b$. This variational characterization provides the basis for the recursive update rules of CG.

At each iteration, the next iterate is obtained by updating the current solution along a search direction $p_{(k)}$,
\begin{align}
    x_{(k+1)} = x_{(k)} + \alpha_{(k)} p_{(k)},
\end{align}
where the scalar $\alpha_{(k)}$ is chosen to minimize the quadratic function along $p_{(k)}$. The search directions are constructed to be $A$-conjugate. Starting from the initial residual $r_{(0)} = b - Ax_{(0)}$, the first direction is set as $p_{(0)} = r_{(0)}$. Enforcing a Galerkin condition yields the coefficient
\begin{align}
    \alpha_{(k)} = \frac{\left(r_{(k)}\right)^T r_{(k)}}{\left(p_{(k)}\right)^T A p_{(k)}}.
\end{align}

The residuals evolve according to
\begin{align}
    \begin{aligned}
        r_{(k+1)} &= r_{(k)} - \alpha_{(k)}Ap_{(k)},
    \end{aligned}
    \label{eq:Residual}
\end{align}
and subsequent directions are generated by the standard two-term recurrence,
\begin{align}
    p_{(k+1)} = r_{(k+1)} + \frac{\left(r_{(k+1)}\right)^T r_{(k+1)}}{\left(r_{(k)}\right)^T r_{(k)}}p_{(k)}.
\end{align}
In particular, \eqref{eq:Residual} shows that $r_{(k+1)}$ is a linear combination of the previous residual $r_{(k)}$ and the direction $Ap_{(k)}$. When $x_{(0)} \neq 0$, the initial residual $r_{(0)}$ encodes prior information about the solution, allowing CG to converge in fewer iterations depending on the quality of the warm start, up to effects caused by finite-precision arithmetic.

From the spectral convergence of iterative methods, the initial error can be decomposed into eigenmodes of $A$, and the components associated with small eigenvalues typically decay most slowly across iterations~\cite{doi:10.1137/S1052623400369235,NEURIPS2020_288cd256,Egger2021warmstartingquantum,10234243,Zou_2025}. From another point of view, Ref.~\cite{10.1145/3652510} provided a systematic mapping study of warm-start techniques in quantum computing and indicated potential advantage in quantum-to-classical warm-start algorithms. Having that said, suppressing these low-energy components in the initial error reduces the effective degrees of a polynomial required for convergence and can significantly reduce the number of CG iterations. The iterative structure of CG naturally exposes intermediate solution scales through the residuals $\{r_{(k)}\}$.

\fig{qacg} illustrates the conceptual mechanism of QACG. \figg{qacg}{a} depicts the two-stage procedure. The quantum stage employs an HHL to prepare an approximate solution state $\ket|\tilde{x}>$ by encoding the inverse of the system matrix within a low-energy subspace. This stage does not aim to fully solve the linear system; instead, it produces an informed initial guess. The classical stage then applies warm-start CG starting from $x^{(0)} = \ket|\tilde{x}>$, iteratively computing step sizes, updating the solution, residual, and search directions, and terminating once the desired accuracy is reached. By shifting the starting point closer to the solution, the iteration depends on an effectively improved conditioning and converges faster than a standard initialization.

\figg{qacg}{b} shows how to read the accelerated CG search on a quadratic objective surface. The contour lines represent level sets of the objective function, and the arrow indicates the trajectory along conjugate directions. The initialized states $\ket|\tilde{x}>$ define an initial point located in the low-energy region, so the subsequent classical iterations traverse a shorter path to the minimum. This reduced traversal reflects the acceleration gained by combining quantum state preparation with classical optimization.

This perspective motivates the use of HHL as a preconditioning mechanism rather than a stand-alone solver. The method constructs a physically informed initial guess that reshapes the region explored by CG and accelerates convergence within a quantum--HPC workflow. To obtain an explicit classical representation of the vector corresponding to $\ket|\tilde{x}>$ that accelerates the CG search, we incorporate a modified HHL routine.

\subsection{Initial guess via the HHL algorithm}

The original statement of quantum linear systems problem (QLSP) is defined as the task of preparing the quantum state~\cite{Childs_2017}
\begin{align}\label{eq:QLSA}
    \ket|x> = \hat{A}^{-1} \ket|b>,
\end{align}
where $\hat{A}^{-1}$ denotes a normalized implementation of $A^{-1}$ acting on the input state $\ket|b>$. With QACG, we approximate the inversion rather than reproducing the full inverse required for an end-to-end quantum linear solver, and we use the result only to warm-start CG. In particular, we apply a spectral filter that inverts only a low-energy window (below a cutoff $\lambda_{\text{cutoff}}$), yielding an initialization that is sufficient to reduce the subsequent classical iteration count.

For the residual criterion $r_{(k)}\le\varepsilon$ that would otherwise require a Krylov space of order $j$, we truncate the inversion at order $i$ by focusing on a low-energy subspace. For an initial residual $r_{(0)}=b-Ax_{(0)}$, the truncated Krylov subspace is denoted as
\begin{align}\label{eq:KrylovQuantum}
  \mathcal{K}_i = \text{span}\left\{
    r_{(0)}, Ar_{(0)}, A^2 r_{(0)}, \ldots, A^{i-1} r_{(0)}
  \right\}.
\end{align}
With truncated Krylov subspaces $\mathcal{K}_i$, we define the filtered inverse operator
\begin{align}\label{eq:SpectralDecomposition}
    \begin{aligned}
        \tilde{A}^{-1}
        &= \sum_j \tilde{f}(\lambda_j)
           \ket|\lambda_j>\bra<\lambda_j|, \\
        \tilde{f}(\lambda,\lambda_{\text{cutoff}})
        &= \begin{cases}
            1/\lambda, & \lambda \le \lambda_{\text{cutoff}}, \\
            0, & \lambda > \lambda_{\text{cutoff}},
        \end{cases}
    \end{aligned}
\end{align}
where $\lambda_{\text{cutoff}}$ is an eigenvalue threshold defining the low-energy subspace. Applying $\tilde{A}^{-1}$ produces the state
\begin{align}\label{eq:SpectralInitialization}
    \begin{aligned}
        \ket|\tilde{x}> &= \tilde{A}^{-1}\ket|b> \\
        &= \tilde{f}(\lambda)\ket|\lambda>\braket<\lambda|b> \\
        &= c\tilde{f}(\lambda)\ket|u>,
    \end{aligned}
\end{align}
where $c=\braket<\lambda|b>$ and $\ket|u>$ denotes the normalized filtered component. Note that eigenvalue inversion with amplitude amplification can be improved to $\mathcal{O}(\kappa\polylog(\kappa))$ using Richardson extrapolation, whereas the straightforward implementation generally requires $\mathcal{O}(\kappa^2)$. The technique is detailed in \app{eigenvalue_inversion_using_Richardson_extrapolation}.

The HHL algorithm~\cite{harrow2009quantum} prepares this state through the composite unitary
\begin{align}
    \begin{aligned}
    \tilde{U}_\text{HHL}\ket|0>\ket|0>\ket|0>
      &= \left(\tilde{U}_\text{QPE}^\dagger
          \tilde{U}_{\text{INV}}
          \tilde{U}_\text{QPE}
          \tilde{U}_\text{INIT}\right)
         \ket|0>\ket|0>\ket|0> \\
      &= \ket|0>\ket|u>\left(
          \sqrt{1 - |c\tilde{f}(\lambda)|^2}\ket|0>
          + c\tilde{f}(\lambda)\ket|1>
         \right),
    \end{aligned}\label{eq:HHL}
\end{align}
where the ancilla state $\ket|1>$ marks successful inversion with probability proportional to $|c\tilde{f}(\lambda)|^{2}$.

To amplify this component, one applies reflections in the amplitude amplification routine:
\begin{align}
  \tilde{U}_\text{AA}\ket|0>\ket|0>\ket|0>
  &= (U_s \tilde{U}_\omega)^k
     \tilde{U}_\text{HHL}
     \ket|0>\ket|0>\ket|0> \\
  &\simeq
     \tilde{f}(\lambda)\ket|0>\ket|u>\ket|1>,
\end{align}
with
\begin{align}\label{eq:AmplitudeAmpliticationOperator}
    \begin{aligned}
        \tilde{U}_\omega
            &= \mathbb{I}-2\ket|\beta>\bra<\beta|
               \coloneqq \tilde{U}_\text{HHL}U_0\tilde{U}_\text{HHL}^\dagger, \\
        U_s
            &= 2\ket|s>\bra<s|-\mathbb{I}
               \coloneqq
               2(\mathbb{I}\otimes\mathbb{I}\otimes\ket|1>\bra<1|)
               -\mathbb{I}, \\
        U_0
            &= \mathbb{I}-2\ket|\bm{0}>\bra<\bm{0}|,
    \end{aligned}
\end{align}
which act as the oracle and diffusion operators.

\alg{QACG} summarizes the complete procedure of QACG. The quantum condition number $\kappa'$ is a parameter that determines the truncation order of subspaces to produce the quantum initial guess. Hence $\kappa'$ plays an important role in the elapsed time to complete the whole procedure. Furthermore, we assume the \textsf{ClassicalDecode} function as an idealized interface that extracts task-relevant classical information from $\ket|\tilde{x}>$ to warm-start CG.

\begin{algorithm}[H]
    \caption{Quantum-accelerated conjugate gradients (QACG)}
    \label{alg:QACG}
    \DontPrintSemicolon
    \KwIn{
    Hermitian matrix $A \in \mathbb{C}^{n\times n}$,
    right-hand side $b \in \mathbb{C}^n$,
    tolerance $\varepsilon$,
    quantum condition number $\kappa'$
    }
    \KwOut{Approximate solution $x_{(k)}$ to $Ax = b$}
    Perform HHL with quantum condition number $\kappa'$, prepare a quantum state proportional to $\tilde{A}^{-1}\ket|b>$, and use it to initialize the classical iterate:
    
    \[
    \ket|\tilde{x}> =
    \frac{\tilde{A}^{-1}\ket|b>}
    {\left\lVert \tilde{A}^{-1}\ket|b> \right\rVert},
    \quad
    x_{(0)} = \textsf{ClassicalDecode}\left(\ket|\tilde{x}>\right).
    \]
    
    $r_{(0)} = b - A x_{(0)}$\;
    $p_{(0)} = r_{(0)}$\;
    $k = 0$\;
    \While{not close enough to solution}{
        $\alpha_{(k)} = \frac{\left(r_{(k)}\right)^T r_{(k)}}{\left(p_{(k)}\right)^T A p_{(k)}}$\;
        $x_{(k+1)} = x_{(k)} + \alpha_{(k)} p_{(k)}$\;
        $r_{(k+1)} = r_{(k)} - \alpha_{(k)}Ap_{(k)}$\;
        $p_{(k+1)} = r_{(k+1)} + \frac{\left(r_{(k+1)}\right)^T r_{(k+1)}}{\left(r_{(k)}\right)^T r_{(k)}}p_{(k)}$\;
        $k = k + 1$\;
    }
\end{algorithm}

We emphasize that QACG does not rely on an exact classical readout (full state tomography) of $\ket|\tilde{x}>$. Instead, it suffices to obtain a classical approximation that is accurate enough for warm-starting CG, i.e., within a controlled approximation error that leads to a reduced initial residual.

Instead of assuming an exact classical readout of the quantum state $\ket|\tilde{x}>$, we assume that task-relevant classical information sufficient for warm-starting CG can be obtained via classical approximation of the quantum state. In particular, quantum states $\ket|\tilde{x}>$ with restricted eigenmodes (frequencies) may admit efficient classical approximations with polynomial computational cost and controllable approximation error, for example by using tensor network representations~\cite{Miyamoto2023} or via polynomial expansions~\cite{rendon2026preconditionedmultivariatequantumsolution,rendon2025exponentiallyimprovedconstantquantum,huang2025fourierspacereadoutmethod}. Such approximations need only preserve the components relevant to reducing the initial residual of CG, rather than reproducing the full quantum state exactly.
\section{Performance estimation for quantum-HPC platforms}
\label{sec:performance_estimation_for_quantum-HPC_platforms}
In this section, we present estimated runtimes and resource requirements. The defail calculation is obtained by time complexity analysis derived by \thm{time-cg}, \thm{time}, and \thm{time-qacg} in \app{time_complexity_analysis}. For CG, the runtime is estimated using publicly available results from the HPCG benchmark on the SoftBank Corp.\ HPC platform CHIE-4. For HHL and QACG, we evaluate fault-tolerant quantum costs at the logical level with a STAR architecture quantum devicie~\cite{PRXQuantum.5.010337}. Specifically, we assume an QEC cycle time $\tau \in \{1\,\mu\text{s},\,1\,\text{ns}\}$ to examine both conservative near-term and optimistic future hardware scenarios. The code distance and the RUS steps are fixed at $d=7$ and $r=2$, which we assume are sufficient to suppress logical errors to a level compatible with the target circuit depth considered in this work. Further more, we set the truncation order $m=10$ for FSL and $p=1$ for the parallel polynomial evaluation. The resource estimation method is described in \app{method_of_resource_estimation}. Additional analysis and the source data appear in \app{performance_analysis} and \app{source_data}.

\begin{figure*}[t]
    \centering
    \includegraphics[width=\linewidth]{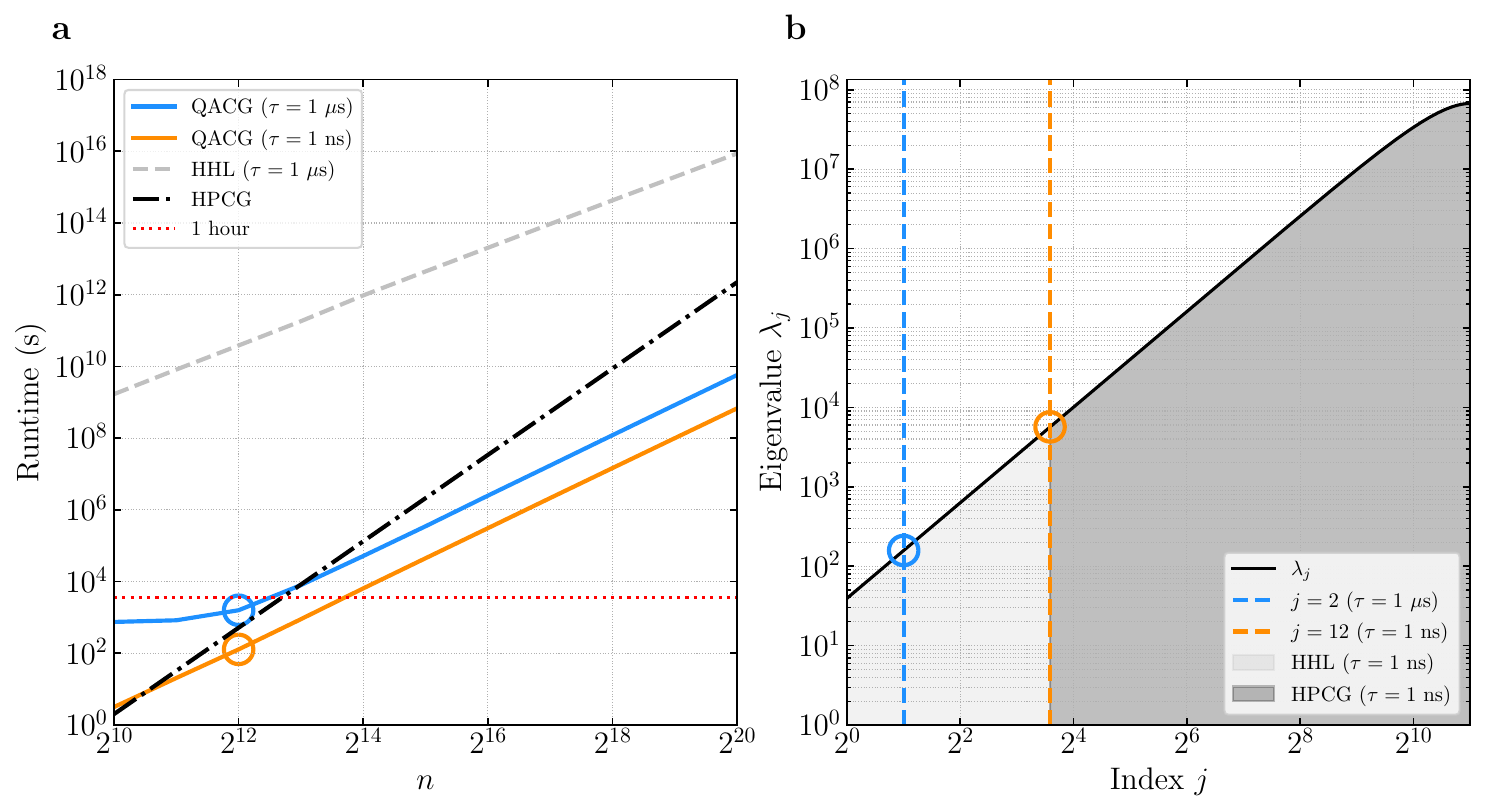}
    \caption{Expected runtime for CG, HHL, and QACG applied to the 3D Poisson equation on an integrated quantum--HPC platform.
    \textbf{a}, Expected runtime as a function of problem size $n$ for classical HPCG, HHL, and QACG with QEC cycle times $\tau=1\,\mu\text{s}$ and $\tau=1\,\text{ns}$. While HPCG remains faster than QACG at $n\approx 2^{12}$ for $\tau=1\,\mu\text{s}$, QACG with $\tau=1\,\text{ns}$ achieves a lower expected runtime beyond this scale. HHL remains significantly slower due to its dependence on the full condition number and large fault-tolerant overheads. 
    \textbf{b}, Eigenvalue spectrum $\{\lambda_j\}$ of the Poisson operator and the optimized spectral initialization window selected by QACG at $n=2^{12}$. For $\tau=1\,\mu\text{s}$, the window remains at the spectral origin ($j=2$), providing negligible acceleration. In contrast, for $\tau=1\,\text{ns}$, QACG initializes a broader low-energy subspace ($j=12$), enabling more effective suppression of classical residuals and a corresponding reduction in expected runtime.}
    \label{fig:3d_runtime_theor}
\end{figure*}

A distinctive feature of QACG is that effective conditioning is decomposed into two parts, parameterized by a cutoff $\lambda_{\text{cutoff}}$: the quantum stage targets the low-energy window $[\lambda_{\min},\lambda_{\text{cutoff}})$, and the classical stage refines the remaining spectral content in $[\lambda_{\text{cutoff}},\lambda_{\max}]$. We denote the corresponding effective condition numbers by $\kappa'$ and $\kappa''$, respectively. By tuning this decomposition, one controls the quality of the initial solution provided by the quantum subroutine and thereby balances the computational burden between quantum and classical resources. Within this model, the allocation of the conditioning between $\kappa'$ and $\kappa''$ is determined by minimizing the total runtime using the COBYQA~\footnote{{COBYQA} {V}ersion 1.1.3; https://www.cobyqa.com; accessed 19 February 2026}. We set the tolerance for CG, HHL, and QACG equally as $\varepsilon=10^{-6}$. In the rest of the work, we restrict attention to a serial execution model in which quantum and classical computations are performed sequentially rather than in parallel.

We consider the 3D Poisson equation to have an explicit analysis of large-scale linear systems. The implementation of the quantum circuit is described in \app{quantum_circuit_for_the_poisson_equation}.

\begin{definition}[Finite difference Poisson equation]
    Let $\Omega=(0,1)^3$ be the unit cube. Consider the 3D Poisson equation with periodic boundary conditions
    \begin{align}\label{eq:poisson}
        \begin{aligned}
            -\Delta u(\bm{x}) &= f(\bm{x}) & & \text{for} & \bm{x}&\in\Omega, \\
            u(\bm{x}) &= 0 & & \text{for} & \bm{x}&\in\partial\Omega .
        \end{aligned}
    \end{align}
    Discretize \eqref{eq:poisson} by second-order finite differences on a uniform \emph{cell-centered} grid with spacing $h=1/n$ and $n$ points per coordinate direction, so that the total number of unknowns is $N=n^3$. Denote grid points by
    \begin{align}
        \begin{aligned}
            x_i&=\left(i-\frac12\right)h, &i&=1,\dots,n,\\
            y_j&=\left(j-\frac12\right)h, &j&=1,\dots,n,\\
            z_k&=\left(k-\frac12\right)h, &k&=1,\dots,n,
        \end{aligned}
    \end{align}
    and write $x_{i,j,k}\approx u(x_i,y_j,z_k)$, $f_{i,j,k}\approx f(x_i,y_j,z_k)$. The standard 7-point stencil yields, for interior indices $(i,j,k)$ with ghost-cell eliminations encoding $u=0$ on $\partial\Omega$. Vectorizing the unknowns, the scheme induces a sparse SPD matrix $A\in\mathbb{R}^{N\times N}$ with 7 non-zero entries per row (up to boundaries) such that $Ax=b$ matches \eqref{eq:LSP}.
    \label{defn:poisson}
\end{definition}

\subsection{Runtime estimation}

\begin{figure*}[ht]
    \centering
    \includegraphics[width=\linewidth]{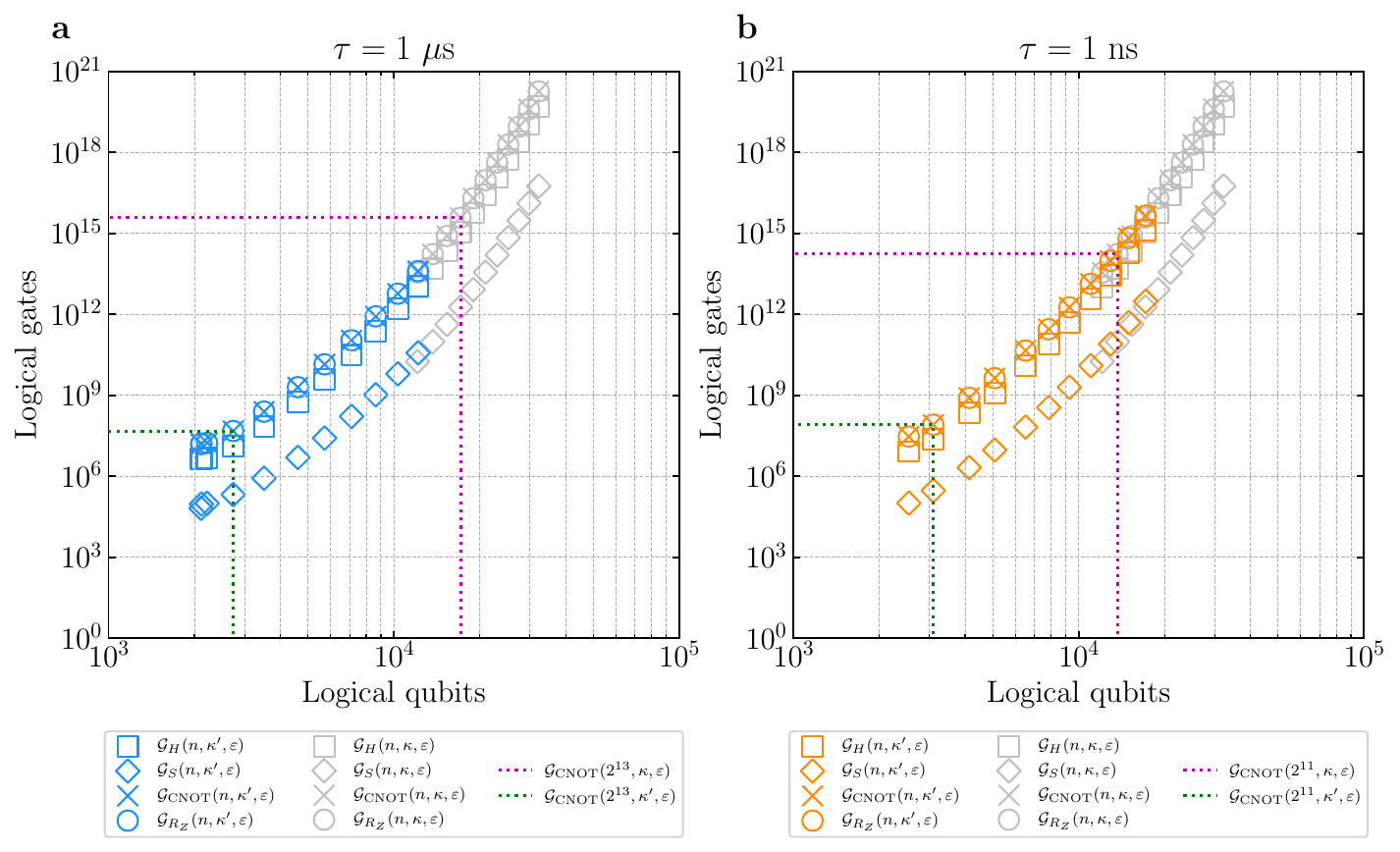}
    \caption{Estimated logical qubits and gates for HHL and QACG for the 3D Poisson equation with grid sizes $n \in \{\ell \in \mathbb{Z} \mid 10 \le \ell \le 20\}$. Gray markers denote HHL evaluated with the full condition number~$\kappa$, whereas colored markers denote QACG using the runtime-optimized quantum condition number~$\kappa'$. \textbf{a},~$\tau = 1~\mu\text{s}$; \textbf{b},~$\tau = 1~\text{ns}$, where $\tau$ is the quantum error correction cycle time. Dotted lines indicate the estimated runtime crossover between QACG and HPCG (benchmark reference). Although logical qubit and gate counts are independent of~$\tau$, reducing $\tau$ shifts the runtime-optimal crossover to smaller problem sizes. Consequently, QACG attains advantage with fewer logical qubits and gates. Across both regimes, QACG lowers logical gate counts by several orders of magnitude relative to HHL while requiring substantially fewer logical qubits.}
    \label{fig:3d_resource_theor}
\end{figure*}

\fig{3d_runtime_theor} summarizes the expected runtime of the proposed QACG, compared against classical CG executed on an HPC system and HHL. \figg{3d_runtime_theor}{a} shows the expected runtime as a function of problem size $n$ for two representative QEC cycle times, $\tau = 1\,\mu\text{s}$ and $\tau = 1\,\text{ns}$. In the conservative setting $\tau = 1\,\mu\text{s}$, the runtime of HHL grows rapidly with $n$, reflecting its dependence on the full condition number $\kappa$. QACG substantially reduces this growth by optimizing the effective quantum condition number $\kappa'$, almost at the crossover scale $n \approx 2^{12}$, while classical CG on HPC remains faster than QACG.

In contrast, when the QEC cycle time is reduced to $\tau = 1\,\text{ns}$, the balance shifts qualitatively. In this regime, QACG achieves a lower expected runtime than classical CG starting at approximately $n \approx 2^{12}$, and the separation between the two methods widens with increasing $n$. These trends highlight the critical role of the QEC cycle time in enabling runtime advantages: only when the quantum subroutine can be executed sufficiently fast does the redistribution of conditioning between quantum and classical resources translate into a net speedup.

\figg{3d_runtime_theor}{b} shows the eigenvalue spectrum $\{\lambda_j\}$ of the Poisson operator, together with the spectral initialization windows employed by QACG at $n = 2^{12}$. The vertical dashed lines indicate the index $j$ up to which the low-energy subspace is handled by the quantum subroutine. For $\tau = 1\,\mu\text{s}$, the optimal window corresponds to $j=2$, which lies to the spectral origin under periodic boundary conditions. In this case, the quantum initialization covers only a negligible portion of the spectrum, explaining why QACG does not outperform classical CG in \figg{3d_runtime_theor}{a}. By contrast, for $\tau = 1\,\text{ns}$, the optimal window expands to $j=12$, indicating that a substantially larger low-energy subspace can be handled on a quantum computer. The fast generation of an initial guess in this low-energy sector enables a more effective reduction of the classical iteration count, leading directly to the observed runtime improvement.

These results suggest that the reduction in runtime achieved by QACG is driven jointly by fast QEC cycle times and by the ability to initialize a nontrivial low-energy spectral band on the quantum processor. Notably, the onset of this reduction occurs at $n \approx 2^{12}$ in the $\tau = 1\,\text{ns}$ regime, corresponding to a linear system with $N \approx 2^{36}$ unknowns. At this scale, classical HPC systems begin to face severe memory-access constraints, whereas an integrated quantum--HPC platform can exploit early FTQC devices to complement state-of-the-art classical resources.

\subsection{Resource estimation}

\begin{figure*}[t]
    \centering
    \includegraphics[width=\linewidth]{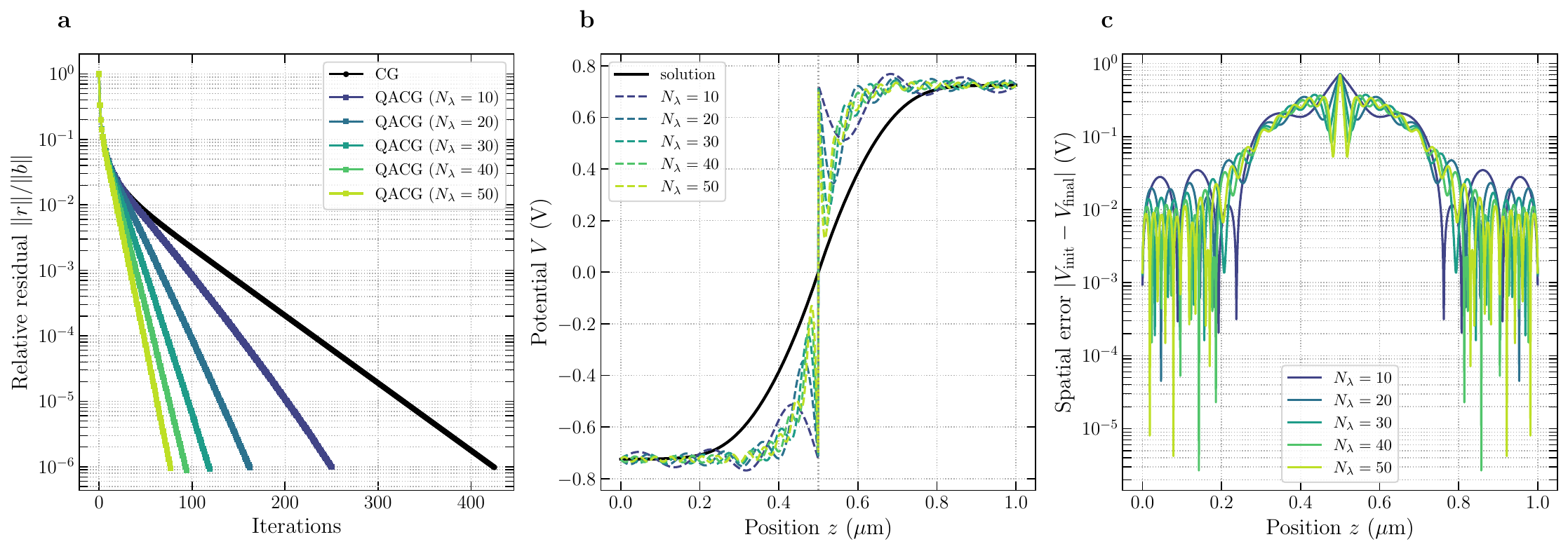}
    \caption{Effect of spectral initialization on CG for the 1D p--n diode Poisson problem.
    \textbf{a}, Relative residual as a function of iteration number for cold-start CG and CG with spectral initialization using $N_\lambda \in \{10,20,30,40,50\}$ eigenvalues with a grid size $n = 2^{10}$.
    \textbf{b}, Potential profiles obtained with spectral initialization compared with the converged solution; dashed curves denote the initial guess $V_{\text{init}}(z)$ from~\eqref{eq:SpectralInitialization}. 
    \textbf{c}, Spatial error $|V_{\text{init}}(z) - V_{\text{final}}(z)|$ of the initial guess, showing oscillatory structure concentrated near the p--n junction.}
    \label{fig:pn_diode_cg}
\end{figure*}

In \fig{3d_resource_theor}, we show the logical resource requirements for executing HHL and runtime-optimized QACG on the STAR architecture with \defn{STAR} in \app{time_complexity_analysis}. Specifically, we plot the number of logical qubits and logical gate counts required to implement HHL with $\kappa'=\kappa$ and QACG with an optimized quantum condition number $\kappa'$, for two representative QEC cycle times, $\tau = 1\,\mu\text{s}$ and $\tau = 1\,\text{ns}$. The logical gate counts $\mathcal{G}_H$, $\mathcal{G}_S$, $\mathcal{G}_{\text{CNOT}}$, and $\mathcal{G}_{R_Z}$ are evaluated using the expressions with \defn{gate-hhl-aa} in \app{method_of_resource_estimation}, with $\kappa'$ selected by COBYQA through minimization of the expected runtime.

\figg{3d_resource_theor}{a} corresponds to the conservative setting $\tau = 1\,\mu\text{s}$. Over the range of problem sizes considered, the gray points represent HHL, while the light-blue points correspond to QACG. As expected, the logical qubit and gate requirements of HHL increase rapidly with $n$, reflecting the direct dependence of the algorithm on the full condition number $\kappa$. In contrast, QACG exhibits a markedly reduced growth rate once the runtime-optimal $\kappa'$ is substantially smaller than $\kappa$. This reduction arises because only a restricted low-energy spectral band is treated by the quantum subroutine, while the remaining spectral components are delegated to classical CG refinement on the HPC system.

At representative scales highlighted by the dotted reference lines at $n=2^{13}$, restricting the quantum subroutine to the optimized low-energy sector reduces the CNOT count by from $3.8 \times 10^{15}$ to $4.8 \times 10^{7}$ relative to HHL, while lowering the logical qubit requirement from $1.7 \times 10^4$ to $2.7 \times 10^3$. Across the full range shown, QACG consistently operates in a regime where both the qubit count and total logical gate counts are far below those required for executing HHL with $\kappa'=\kappa$. Among the individual gate types, the $H$, $R_Z$, and CNOT counts exhibit similar scaling behavior, consistent with their dominant dependence on $\kappa'$, whereas the $S$-gate count remains subdominant throughout.

\figg{3d_resource_theor}{b} shows the corresponding resource estimates for a faster QEC cycle, $\tau = 1\,\text{ns}$. As expected, the logical gate and qubit counts themselves do not depend explicitly on $\tau$; however, the reduced cycle time shifts the runtime-optimal crossover between QACG and classical HPCG to smaller problem sizes i.e. $n=2^{11}$, as indicated by the dotted reference lines. In particular, the crossover is reached at lower values of the logical qubit count than in the $\tau = 1\,\mu\text{s}$ case. Although the faster QEC cycle favors larger optimal values of $\kappa'$ in the runtime, QACG achieves its runtime advantage at a smaller problem size with fewer logical qubits. In the optimistic QEC cycle time, the logical qubit required by QACG reduced from $1.4 \times 10^4$ to $3.1 \times 10^3$ and logical CNOT counts reduced approximately from $1.7 \times 10^{14}$ to $8.2 \times 10^7$ relative to HHL.

Finally, we emphasize that the substantial reduction in logical qubits and gate counts achieved by QACG is obtained by explicitly assuming that the remaining spectral components are processed on the HPC system via classical CG. This suggests that the optimization of the spectral initialization should be calibrated jointly against the capabilities of both the quantum and the HPC platform. Such co-design is essential for maximizing overall resource utilization and performance in an integrated quantum--HPC platform.

\subsection{Numerical simulation for the 1D Poisson equation}

To verify QACG (\alg{QACG}), we numerically investigate the effect of spectral initialization on CG convergence for a 1D semiconductor device simulation. Here, spectral initialization serves as a classical surrogate for the output of the HHL algorithm~\cite{harrow2009quantum}, which encodes an approximate solution of $Ax=b$ in the amplitudes of a quantum state through phase estimation with a limited number of ancilla qubits.

We consider the 1D Poisson equation for an abrupt p--n junction diode at thermal equilibrium~\cite{beckers2022robustsimulationpoissonsequation}:
\begin{equation}
    \frac{d^2 V}{dz^2} = -\frac{\rho(V)}{\varepsilon_{\text{Si}}},
\end{equation}
where $V(z)$ is the electrostatic potential, $\varepsilon_{\text{Si}} = 1.05 \times 10^{-12}\,(\text{F/cm})$ is the silicon permittivity, and the charge density under Boltzmann statistics is given by
\begin{equation}
    \rho(V) = e\left(N_D - N_A - n_i e^{V/U_T} + n_i e^{-V/U_T}\right),
\end{equation}
with $U_T = k_B T/e$ denoting the thermal voltage, $n_i$ the intrinsic carrier density, and $N_A$, $N_D$ the acceptor and donor doping concentrations, respectively.

The simulation parameters are: temperature $T = 300\,(\text{K})$, symmetric doping $N_A = N_D = 10^{16}\,(\text{cm}^{-3})$, bandgap energy $E_g = 1.12\,(\text{eV})$, device length $L = 1\,(\mu\text{m})$, and $n = 2^{10}$ interior grid points. The discretized system yields a tridiagonal Jacobian matrix $A \in \mathbb{R}^{n \times n}$ that is SPD.

The HHL algorithm is assumed to construct an approximate solution by projecting onto the eigenspace of $A$ via quantum phase estimation. When phase estimation uses $n_{\lambda}$ ancilla qubits, only $N_{\lambda}=2^{n_{\lambda}}$ distinct eigenvalue bins are resolvable, effectively truncating the spectral representation to a limited number of eigenmodes. We model this spectral filtering classically by computing the truncated spectral inverse with \eqref{eq:SpectralInitialization}. This initialization $x_{(0)}$ serves as the starting point for subsequent CG iterations.

\fig{pn_diode_cg} presents the convergence behavior of CG with varying degrees of spectral initialization. \figg{pn_diode_cg}{a} shows the relative residual $\|r\|/\|b\|$ as a function of iteration count. Starting from a cold (zero) initialization, CG requires 425 iterations to reach a relative tolerance of $\varepsilon = 10^{-6}$. In contrast, spectral initialization with $N_\lambda \in \{10,20,30,40,50\}$ eigenvalues reduces the iteration count to $\{250,162,119,94,77\}$, respectively, demonstrating a consistent improvement with increasing spectral resolution.

\figg{pn_diode_cg}{b} compares the spectrally initialized potential profiles with the converged solution. While all initializations capture the qualitative shape of the built-in potential, larger $N_\lambda$ yields closer approximations to the final solution. \figg{pn_diode_cg}{c} quantifies the spatial distribution of the initialization error $|V_{\text{init}} - V_{\text{final}}|$. The error exhibits oscillatory structure, with the largest deviations occurring near the p--n junction where the potential gradient is steepest. This behavior is consistent with the fact that low-frequency eigenmodes (small $\lambda_k$) capture smooth, long-wavelength features, while the sharp junction transition requires higher-frequency components for accurate representation.
\section{Conclusion}
\label{sec:conclusion}
We have presented an integrated quantum--HPC algorithm for solving large-scale linear systems using CG, where a fault-tolerant quantum subroutine is used only to generate an initial guess.
We formulated this strategy as QACG and showed, under explicit hardware assumptions, that quantum resources can be effectively utilized without executing an end-to-end quantum linear solver while improving performance relative to a benchmark-calibrated classical baseline.

Our central observation is that the benefit of quantum acceleration arises not from replacing CG, but from decomposing the workload across classical and quantum resources according to the spectral structure of the linear operator.
By focusing on a low-energy spectral window to the quantum initialization and delegating the remaining spectral content to classical CG, QACG exhibits a crossover regime in which both the estimated runtime and the quantum resources are reduced.
In our analysis, this crossover depends jointly on the problem size, the quantum condition number $\kappa'$, the classical effective condition number $\kappa''$, and the assumed fault-tolerant clock speed, rather than on any single parameter.

Based on time complexity bounds and runtime/resource estimates based on the STAR architecture model, we found parameter regimes in which QACG outperforms the HPCG-calibrated classical baseline for the 3D Poisson problem.
Importantly, the improvement does not rely on implementing a full inverse on the quantum device.
Instead, restricting the quantum computation to a narrow spectral window reduces the logical gate counts and qubit requirements while providing a warm start that decreases the number of costly classical CG iterations.
These results provide a concrete example of a regime beyond monolithic architecture in which classical and quantum resources cooperate to expand performance under the stated assumptions.

Several limitations should be emphasized. Notably, we assumed that the quantum subroutine can hand over task-relevant classical information derived from the quantum state with negligible overhead.
This assumption isolates the scaling behavior of the algorithmic decomposition but is not generally satisfied in realistic settings.
Achieving practical impact will therefore require efficient interfaces between quantum and classical computations, including readout strategies that extract only  limited information or restricted modes rather than reconstructing the full quantum state.
Promising directions include tighter integration between quantum phase estimation and Krylov-based workflows and hybrid schemes in which quantum routines provide spectral information that can be used to improve classical convergence.

Looking forward, the framework developed here suggests a pathway for extending the reach of classical HPC by integrating fault-tolerant quantum devices as cooperative accelerators within an integrated quantum--HPC system.
Beyond the Poisson equation, it will be important to study broader classes of problems, including more complex PDEs.
Incorporating communication costs, synchronization overheads, and alternative fault-tolerant architectures will be necessary to refine quantitative predictions.
More broadly, our results support the view that early, practical impact of early-stage fault-tolerant quantum computing in scientific and industrial workloads may come from hybrid schemes that align quantum subroutines with the specific bottlenecks of classical algorithms, rather than from fully quantum end-to-end solvers.

\begin{acknowledgments}
    The authors are grateful to Yosuke Komiyama for valuable discussions and insightful comments that greatly contributed to this work.
\end{acknowledgments}




\section*{Data availability statement}
All study data are included in \app{source_data}.

\section*{Author contributions}
S.M. conceived the initial theoretical framework and performed the numerical simulations. S.M. and Y.S. developed the theoretical concepts. Y.S. supervised the project. All authors contributed substantially to the work.

\clearpage
\onecolumngrid
\appendix

\appendixtoctrue
\appendixtableofcontents
\vspace{3em}

We provide supplementary information to produce the results presented in the main text.
In \app{quantum_circuit_for_the_poisson_equation}, we describe the quantum circuit implementation for simulating the 3D finite-difference Poisson equation.
In \app{eigenvalue_inversion_using_Richardson_extrapolation}, we present the eigenvalue inversion technique employed in the implementation.
In \app{method_of_resource_estimation}, we detail the resource estimation procedure.
In \app{time_complexity_analysis}, we derive time complexity of the CG, HHL, and QACG.
In \app{performance_analysis} we provide additional analysis and supporting evidence for the scaling behavior and resource estimates of QACG.
In \app{source_data} we provide source data to reproduce the result in the main text.

\section{Quantum circuit for the 3D Poisson equation}
\label{app:quantum_circuit_for_the_poisson_equation}

We explain the quantum circuit implementation of the Harrow–Hassidim–Lloyd (HHL) algorithm~\cite{harrow2009quantum} used to simulate the 3D Poisson equation. The construction is tailored to the QACG framework, in which the spectrum is truncated above a prescribed threshold to define an effective quantum condition number. The same circuit structure can also be employed within the standard HHL algorithm without spectral truncation; the distinction lies in the eigenvalue filtering step rather than in the overall architecture of the circuit.

\begin{figure*}[htb]
  \centering
  \begin{tikzpicture}
    \node[anchor=north west,inner sep=0] (A)
      {\includegraphics[width=0.90\linewidth]{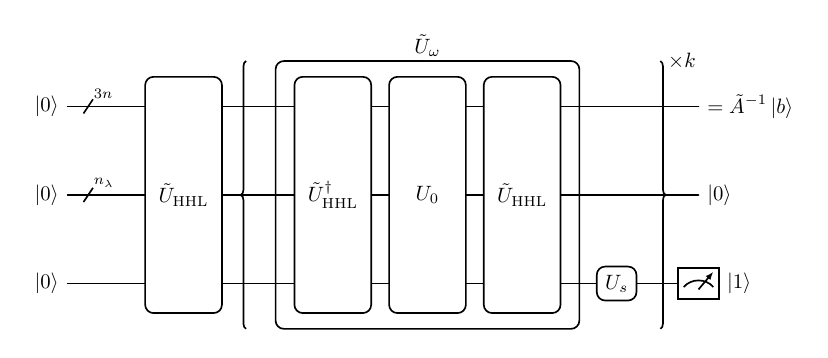}};
    \node[anchor=north west,font=\bfseries\large]
      at ([xshift=5mm,yshift=-5mm]A.north west) {a};
  \end{tikzpicture}
  \label{fig:AA}

  \begin{tikzpicture}
    \node[anchor=north west,inner sep=0] (B)
      {\includegraphics[width=0.80\linewidth]{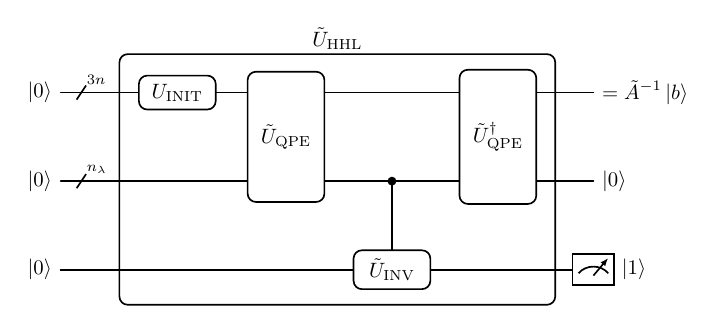}};
    \node[anchor=north west,font=\bfseries\large]
      at ([xshift=-8mm,yshift=-5mm]B.north west) {b};
  \end{tikzpicture}
  \label{fig:HHL}

  \caption{Quantum circuits used to construct an initial guess via spectral initialization for the 3D Poisson equation. \textbf{a}, Amplitude amplification $(U_s\tilde{U}_\omega)^k$ boosts the success probability of the solution state. \textbf{b}, HHL operator composed of state preparation, quantum phase estimation, and controlled rotation to implement a spectrally filtered matrix inversion.}
  \label{fig:quantum_circuits}
\end{figure*}

To compute an initial guess $\ket|\tilde{x}>$, one must design and implement quantum circuits for a given $A$ and $b$. In this work, we consider implementing HHL along with amplitude amplification as displayed in \fig{quantum_circuits}. The operator $\tilde{U}_\text{HHL}$ (and $\tilde{U}^\dagger_\text{HHL}$) is composed of state preparation $\tilde{U}_\text{INIT}$, QPE blocks $\tilde{U}_\text{QPE},\tilde{U}^\dagger_\text{QPE}$, and controlled rotation $\tilde{U}_{\text{INV}}$. The circuit $\tilde{U}_\text{INIT}$ prepares $\ket|b>$, QPE estimates eigenvalues of $A$, and the controlled rotation $\tilde{U}_{\text{INV}}$ inverts the spectrum to implement the filtered matrix inversion.

\begin{figure}[htbp]
    \centering
    \includegraphics[width=0.5\linewidth]{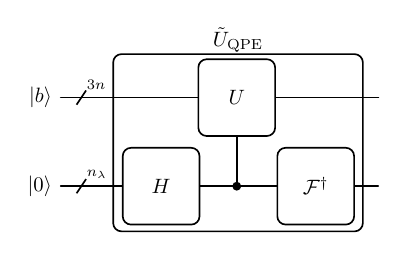}
    \caption{Quantum phase estimation circuit used within HHL. Controlled time-evolution unitaries and an inverse quantum Fourier transform map eigenphases onto a phase register, enabling subsequent controlled rotations that implement spectral inversion.}
    \label{fig:QPE}
\end{figure}

To estimate the eigenvalues of the system matrix and enable spectral inversion, we implement QPE as a core subroutine within the HHL framework. As shown in \fig{QPE}, the QPE block $\tilde{U}_{\text{QPE}}$ consists of an $n_\lambda$-qubit phase register initialized in $\ket|0>$, followed by a layer of Hadamard gates to create a uniform superposition. Controlled applications of the unitary powers $U^{2^{n_\lambda}}$ encode the eigenphases via phase kickback, where $U=\exp(iAt)$ is the time-evolution operator associated with the system matrix $A$. For the specific case of the 3D Poisson equation considered here, the implementation of the controlled powers $U^{2^{n_\lambda}}$ can be substantially abbreviated relative to a generic sparse Hamiltonian. The discretized Laplacian matrix is diagonalizable, implying that its eigenvalues are available in closed form and that the corresponding time-evolution operator can be synthesized without explicitly constructing $2^{n_\lambda}$ sequential applications of $U$. Instead, the controlled $U^{2^{n_\lambda}}$ operations can be realized by reusing a common circuit structure with appropriately rescaled evolution times. As a result, the cost of implementing $U^{2^{n_\lambda}}$ scales with the precision of the phase estimation rather than exponentially with $n_\lambda$. An inverse quantum Fourier transform $\mathcal{F}^\dagger$ is then applied to the phase register to obtain a binary approximation of the eigenvalues $\lambda_j$ of $A$, coherently correlated with the corresponding eigenstates in the data register. This eigenvalue information is subsequently used in the controlled rotation $\tilde{R}_Y$ to implement spectral inversion, completing the HHL procedure.

\begin{figure}[htbp]
    \centering
    \includegraphics[width=0.5\linewidth]{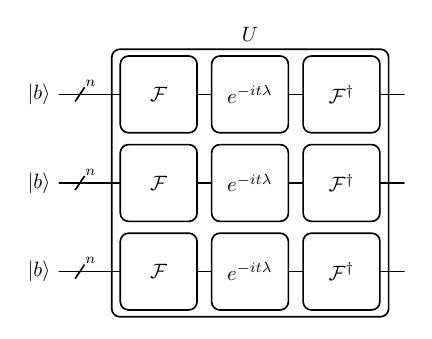}
    \caption{Quantum circuit implementing the unitary for the 3D periodic Laplacian. Quantum Fourier transforms and their inverses are applied to each spatial register, mapping the Laplacian to a diagonal representation. Diagonal phase rotations implement the corresponding eigenvalues, thereby generating the unitary operator associated with the 3D periodic Laplacian.}
    \label{fig:U}
\end{figure}

For periodic boundary conditions, the Laplacian is diagonal in the Fourier basis. In the continuous setting, with Fourier transform $\hat u(\boldsymbol{\kappa})=\mathcal{F}[u](\boldsymbol{\kappa})$, one has
\begin{align}
    \mathcal{F}[\Delta u](\boldsymbol{\kappa})
    = -\|\boldsymbol{\kappa}\|_2^2 \hat u(\boldsymbol{\kappa}),
    \qquad \boldsymbol{\kappa}=(\kappa_x,\kappa_y,\kappa_z).
    \label{eq:fourier_laplacian_continuous}
\end{align}
Thus, $\Delta$ is diagonalized by $\mathcal{F}$ with eigenvalue $-\|\boldsymbol{\kappa}\|_2^2$.

In the discrete periodic case, the 1D circulant Laplacian $A$ in \defn{laplacian} is diagonalized by the discrete Fourier transform (DFT) matrix $F$:
\begin{align}
    A = F^\dagger \Lambda F,
    \qquad
    \Lambda = \text{diag}(\lambda_0,\ldots,\lambda_{n-1}),
    \label{eq:1d_dft_diagonalization}
\end{align}
with eigenvalues (for mode index $k\in\{0,\ldots,n-1\}$)
\begin{align}
    \lambda_k = \frac{4}{h^2}\sin^2\left(\frac{\pi k}{n}\right).
    \label{eq:1d_laplacian_eigs}
\end{align}

The 3D discrete Laplacian on an $n\times n\times n$ periodic grid is separable and can be written as a Kronecker sum
\begin{align}
    A
    = A\otimes I\otimes I + I\otimes A\otimes I + I\otimes I\otimes A,
    \label{eq:3d_kronecker_sum}
\end{align}
so it is diagonalized as
\begin{align}
    A
    = F^\dagger \Lambda F,
    \qquad
    \Lambda=\text{diag}(\lambda_{k_x,k_y,k_z}),
    \label{eq:3d_dft_diagonalization}
\end{align}
where the eigenvalue associated with the Fourier mode $(k_x,k_y,k_z)$ is the additive combination
\begin{align}
    \lambda_{k_x,k_y,k_z}=\lambda_{k_x}+\lambda_{k_y}+\lambda_{k_z}.
    \label{eq:3d_laplacian_eigs}
\end{align}

Consequently, the unitary $U=\exp(iAt)$ admits the spectral implementation~\cite{Benenti_2008,Childs_2020}
\begin{align}
    U
    = \mathcal{F}^\dagger \exp\bigl(i t\Lambda\bigr) \mathcal{F},
    \label{eq:U_spectral_form}
\end{align}
which matches the circuit in \fig{U}: apply $\mathcal{F}^{\otimes 3}$ (a QFT on each spatial register) to move to the Fourier basis, apply a diagonal phase rotation that maps
$|k_x,k_y,k_z\rangle \mapsto e^{it\lambda_{k_x,k_y,k_z}}|k_x,k_y,k_z\rangle$,
and finally apply $(\mathcal{F}^\dagger)^{\otimes 3}$ to return to the computational basis.

The right-hand side vector $\ket|b>$ must be efficiently encoded into a quantum state in order to apply HHL. In this work, we employ the Fourier series loader (FSL) introduced in~\cite{Moosa_2023}, which provides a scalable and structured approach to quantum state preparation based on truncated Fourier expansions and uniformly controlled rotations.

Let $b(x)$ be a real-valued function defined on a uniform grid of size $n$, corresponding to the discretized right-hand side vector
\begin{align}
  \ket|b> = \frac{1}{\|b\|} \sum_{x=0}^{n-1} b(x)\ket|x>.
\end{align}
The FSL approach approximates $b(x)$ by a truncated Fourier series
\begin{align}
  b(x) \approx \sum_{k=-m}^{m} c_k e^{2\pi i k x / n},
  \label{eq:fsl_fourier}
\end{align}
where $c_k$ are Fourier coefficients and $m \ll n$ controls the truncation error. For sufficiently smooth functions $b(x)$, the approximation error decays exponentially in $m$, making the method suitable for vectors arising from PDE discretizations.

The corresponding quantum circuit is shown in \fig{FSL}. The circuit first prepares an intermediate state on an $(m+1)$-qubit register in each direction using a cascaded entangler $U_c$, which maps
\begin{align}
  \ket|0>^{\otimes (m+1)} \mapsto \ket|c>,
\end{align}
where
\begin{align}
  \ket|c>
  &= 2^{-n/2} \sum_{k=0}^{m} c_k \ket|k>
     + 2^{-n/2} \sum_{k=1}^{m} c_{-k} \ket|2^{m+1}-k>.
  \label{eq:fsl_ck_state}
\end{align}
This state coherently encodes both positive- and negative-frequency Fourier coefficients in a compact binary representation.

\begin{figure}[htbp]
  \centering
  \includegraphics[width=0.5\linewidth]{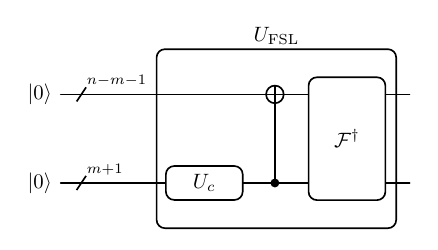}
  \caption{Quantum circuit for the 1D Fourier series loader (FSL). The cascaded entangler $U_c$ prepares truncated Fourier coefficients on an $(m+1)$-qubit register, which is then entangled with the remaining qubits via CNOT ladders and mapped to the spatial domain by an inverse QFT, yielding $\ket|\tilde{b}>$.}
  \label{fig:FSL}
\end{figure}

Next, a cascade of $n-m-1$ CNOT gates entangles the coefficient register with the remaining qubits, mapping
\begin{align}
  \ket|0>^{\otimes (n-m-1)} \otimes \ket|c> \mapsto \ket|c'>,
\end{align}
where the resulting $n$-qubit state is
\begin{align}
  \ket|c'>
  &= 2^{-n/2} \sum_{k=0}^{m} c_k \ket|k>
     + 2^{-n/2} \sum_{k=1}^{m} c_{-k} \ket|2^{n}-k>.
  \label{eq:fsl_cprime_state}
\end{align}
This step distributes the truncated Fourier coefficients across the full computational register.

Finally, an inverse quantum Fourier transform $\mathcal{F}^\dagger$ is applied, yielding
\begin{align}
  \mathcal{F}^\dagger \ket|c'> \approx \ket|b>,
\end{align}
where the approximation error is controlled solely by the truncation order $m$ in 1D.

The circuit depth of the FSL scales as $\mathcal{O}(\log n + 2^m)$ assuming the Fourier coefficients $c_k$ are classically precomputed. Since $m$ is fixed independently of $n$ in our setting, this construction is compatible with large-scale quantum--HPC regimes. Moreover, both the FSL and HHL operate naturally in the spectral domain, so the prepared state integrates directly with the eigenvalue filtering $\tilde{f}(\lambda)$ applied during QPE. As a result, the FSL-prepared vector emphasizes low-frequency components of $\ket|b>$, yielding an initial guess that captures the low-energy structure of the Krylov subspace associated with the discretized Laplacian. Integrating the above procedures yields an initial guess $x_{(0)}$ aligned with the spectral initialization.
\section{Eigenvalue inversion using Richardson extrapolation}
\label{app:eigenvalue_inversion_using_Richardson_extrapolation}

We briefly summarizes the eigenvalue inversion procedure based on Richardson extrapolation introduced in Ref.~\cite{Vazquez_2022}. Let $A$ be a Hermitian matrix with eigenpairs $\{(\lambda_j,\ket|u_j>)\}$ and condition number $\kappa=\lambda_{\max}/\lambda_{\min}$. Within the HHL framework, after quantum phase estimation one obtains an $n_\lambda$-qubit register encoding a fixed-point approximation $\tilde{\lambda}_j$ of each eigenvalue. The inversion step aims to produce an ancilla amplitude proportional to $\tilde{\lambda}_j^{-1}$, while controlling the contribution of this operation to the total algorithmic error.

Rather than implementing a direct controlled rotation with angle proportional to $\tilde{\lambda}_j^{-1}$, Ref.~\cite{Vazquez_2022} approximates the angle function
\begin{align}
    f(x)=\arcsin\left(\frac{c}{x}\right),
\end{align}
where $x$ denotes the integer-valued eigenvalue estimate stored in the register and the constant $c>0$ is chosen such that $0<c/x<1$ over the relevant spectral range. The domain of $x$ is taken to be $[a,2^{n_\lambda}-1]$, where $a$ is a cutoff introduced to control the approximation error near the smallest eigenvalues.

The interval $[a,2^{n_\lambda}-1]$ is partitioned into a union of exponentially growing subintervals,
\begin{align}
    [a_1,a_2]\cup[a_2,a_3]\cup\cdots\cup[a_M,a_{M+1}],
\end{align}
with $a_1=a$ and $a_{\ell+1}=5a_\ell$. The total number of intervals is therefore
\begin{align}
    M=\left\lceil \log_5\left(\frac{2^{n_\lambda}-1}{a}\right)\right\rceil .
\end{align}
On each subinterval, the function $f$ is approximated by a Chebyshev interpolating polynomial of degree $d$. Richardson extrapolation is used to cancel leading-order truncation errors, yielding a piecewise polynomial approximation whose uniform error satisfies
\begin{align}
    \|f-p_f\|_{L^\infty}
    \le \varepsilon_C ,
\end{align}
where $\varepsilon_C$ denotes the target approximation accuracy of the angle function.

Ref.~\cite{Vazquez_2022} shows that choosing
\begin{align}
    \varepsilon_C
    =\frac{\varepsilon}{2(2\kappa^2-\varepsilon)},
\end{align}
together with an eigenvalue register size
\begin{align}
    \label{eq:n_lambda}
    n_\lambda
    =
    3\left\lfloor
    \log\left(
    \frac{2(2\kappa^2-\varepsilon)}{\varepsilon}+1
    \right)
    \right\rfloor ,
\end{align}
ensures that the contribution of the inversion step to the final HHL solution error is at most $\varepsilon$. Under these choices, the polynomial degree required to achieve the above uniform bound can be estimated using analytic continuation of $f$ to a Bernstein ellipse. An explicit bound given in Ref.~\cite{Vazquez_2022} is
\begin{align}
    d=
    \left\lfloor
    \log\left(
    1+
    \frac{16.23
    \sqrt{\ln^2(r)+(\pi/2)^2}
    \kappa(2\kappa-\varepsilon)}
    {\varepsilon}
    \right)
    \right\rfloor ,
\end{align}
where the ellipse parameter $r>1$ depends on the ratio $c/a$. This yields the asymptotic scaling
\begin{align}
    \label{eq:degree}
    d=\mathcal{O}\left(
    \log\left(\frac{\kappa^2}{\varepsilon}\right)
    +\log\log\left(\frac{\kappa^2}{\varepsilon}\right)
    \right).
\end{align}

The polynomial approximation $p_f$ is evaluated reversibly, for example using a reversible Horner scheme~\cite{Haner:2018yea}, to compute an angle $y_j\approx f(\tilde{\lambda}_j)$ from the eigenvalue register. A controlled $R_Y$ rotation acting on an ancilla qubit then implements the standard HHL inversion,
\begin{align}
    \sqrt{1-\frac{c^2}{\tilde{\lambda}_j^{2}}}\ket|0>
    +\frac{c}{\tilde{\lambda}_j}\ket|1>,
\end{align}
up to an error bounded by $\varepsilon$ with the above parameter choices. In this way, the Richardson--Chebyshev construction replaces the direct inversion step by a controlled functional approximation with polylogarithmic dependence on $1/\varepsilon$, while maintaining explicit control over the condition-number dependence of the eigenvalue inversion subroutine, as established in Ref.~\cite{Vazquez_2022}.
\section{Method of resource estimation}
\label{app:method_of_resource_estimation}

In the following, we define some basic properties of the Poisson equation and analyze the space and gate complexity for explicitly estimating the architecture-aware time complexity, including prefactors.

\subsection{Basic properties of the Poisson equation}

\begin{definition}[Discrete Laplacian with periodic boundary conditions]
\label{defn:laplacian}
    Let $\Omega=(0,1)$ and discretize on a uniform grid with spacing $h=1/n$ and grid points
    $x_j=jh$ for $j=0,1,\dots,n-1$. Define the 1D periodic discrete Laplacian
    $L\in\mathbb{R}^{n\times n}$ by
    \begin{align}
    L=\frac{1}{h^2}
        \begin{bmatrix}
         2 & -1 & 0 & \cdots & 0 & -1\\
        -1 &  2 & -1& \ddots &  & 0\\
         0 & -1 & 2 & \ddots & \ddots & \vdots\\
         \vdots & \ddots & \ddots & \ddots & -1 & 0\\
         0 &  & \ddots & -1 & 2 & -1\\
        -1 & 0 & \cdots & 0 & -1 & 2
        \end{bmatrix}.
    \end{align}
    Equivalently, $L$ is circulant with stencil $( -1,+2,-1 )/h^2$ and wrap-around couplings enforcing periodicity.
\end{definition}

\begin{lemma}[Eigenvalues and Fourier diagonalization]
\label{lem:Eigenvalues}
    Let $F$ be the $n\times n$ discrete Fourier transform (DFT) matrix. Then
    \begin{align}
        L = F^\dagger \Lambda F,\qquad
        \Lambda=\text{diag}(\lambda_0,\lambda_1,\dots,\lambda_{n-1}),
    \end{align}
    with eigenvalues
    \begin{align}
        \lambda_k = \frac{4}{h^2}\sin^2\left(\frac{\pi k}{n}\right),
        \qquad k=0,1,\dots,n-1.
    \end{align}
    In particular, $\lambda_0=0$ and $L$ is singular.
\end{lemma}

\begin{proof}
    Consider the standard second-order central finite-difference approximation of the 1D Laplacian on a uniform grid with spacing $h$ under periodic boundary conditions with \defn{laplacian}. In this setting, the discrete Laplacian is represented by a circulant matrix whose first row is given by
    $(2,-1,0,\ldots,0,-1)/h^{2}$ \cite{doi:10.1137/1.9780898717839}.
    It is a well-known result in linear algebra that circulant matrices are diagonalized by the discrete Fourier transform, and that their eigenvalues are obtained as the discrete Fourier transform of the first row \cite{10.1561/0100000006}.
    Applying this general result to the present case, the eigenvalues of the discrete Laplacian are explicitly given by
    \begin{align}
        \begin{aligned}
            \lambda_k
            &= \frac{1}{h^{2}}\left(2 - e^{2\pi i k/N} - e^{-2\pi i k/N}\right) \\
            &= \frac{4}{h^{2}}\sin^{2}\left(\frac{\pi k}{N}\right),
            \qquad k = 0,1,\ldots,n-1,
        \end{aligned}
    \end{align}
    with the corresponding eigenvectors given by the discrete Fourier modes.
\end{proof}

\begin{definition}[Zero-mean subspace and restricted inverse]
\label{defn:zero_mean_subspace}
    Define the zero-mean subspace
    \begin{align}
        \mathcal{H}_0 \coloneqq \left\{ v\in\mathbb{R}^n : \sum_{j=0}^{n-1} v_j = 0 \right\}.
    \end{align}
    We consider the Poisson system $Lx=b$ under the compatibility condition $b\in\mathcal{H}_0$.
    The solution is defined up to an additive constant; we select the zero-mean solution $x\in\mathcal{H}_0$.
    Equivalently, we use the inverse of $L$ restricted to $\mathcal{H}_0$ (or the Moore--Penrose pseudoinverse on $\mathcal{H}_0$).
\end{definition}

\begin{lemma}[Effective condition number on $\mathcal{H}_0$]
    \label{lem:condition}
    On $\mathcal{H}_0$ with \defn{zero_mean_subspace}, the smallest nonzero eigenvalue is $\lambda_1$ and the largest eigenvalue is $\lambda_{\max}$ (for even $n$, achieved at $k=n/2$). Hence the effective condition number of $L$ restricted to $\mathcal{H}_0$ satisfies
    \begin{align}
        \kappa_{\mathcal{H}_0}(L)
        = \frac{\lambda_{\max}}{\lambda_1}
        = \frac{1}{\sin^2(\pi/n)}
        = \Theta(n^2).
    \end{align}
\end{lemma}

\begin{lemma}[Sparsity of the discrete Laplacian in \defn{laplacian}]\label{lem:sparsity}
    For the matrix $A\in\mathbb{R}^{N\times N}$ arising from \defn{poisson} with $N=n^3$ and the 7-point stencil,
    \begin{align}
        \max_{i}\operatorname{nnz}(\text{row } i) &= 7.
    \end{align}
\end{lemma}

\begin{proof}
    Each interior grid point couples to its 6 nearest neighbors plus itself, yielding 7 nonzeros per row. On boundary strata, the counts may be reduced depending on the boundary-condition treatment, while the maximal per-row count remains 7.
\end{proof}

\subsection{Space complexity of HHL}

In this and the following section, we use $n$ to denote the number of qubits, such that the system dimension is $N = 2^{n}$. This convention differs from earlier sections of the paper, where $n$ denotes the number of unknowns per spatial dimension in the discretized Poisson equation. When referring to the Poisson problem, the distinction is explicit through the use of $N = n^{3}$ for the total number of degrees of freedom. Throughout the present discussion of QLSA and HHL complexity, $n$ should therefore be interpreted exclusively as the logarithmic system-size parameter $n = \log_{2} N$, consistent with standard conventions in quantum algorithm analysis.

We estimate the space complexity of HHL by accounting for all qubit registers that may be simultaneously active during the algorithm, with particular attention to ancillary qubits required by arithmetic subroutines and by amplitude amplification.

The system register representing the 3D Laplacian $A$ consists of $3n$ qubits. Quantum phase estimation (QPE) introduces an $n_\lambda$-qubit eigenvalue register, where $n_\lambda$ is determined by the precision required to resolve eigenvalues down to $\Theta(1/\kappa)$. In addition, HHL uses a single flag qubit to implement the eigenvalue-dependent rotation and postselection.

The dominant ancillary overhead arises from the eigenvalue-inversion step, which we implement using the parallel piecewise-polynomial evaluation circuit of~\cite{Haner:2018yea}. As stated in that resource estimate, the polynomial-evaluation circuit requires
\begin{align}
    \mathcal{S}_{\text{pp}}(n_\lambda,d,M)
    =
    (d+1)n_\lambda + \left\lceil\log_2 M\right\rceil + 1
\end{align}
qubits in total, including the $n_\lambda$-qubit input register holding the eigenvalue. Therefore, beyond the eigenvalue register itself, the arithmetic introduces an additional workspace of
\begin{align}
    d n_\lambda + \left\lceil\log_2 M\right\rceil + 1
\end{align}
qubits. Here $d$ is the polynomial degree and $M$ is the number of piecewise intervals, both fixed by the inversion accuracy requirements and already determined by the gate-complexity analysis.

A further source of ancillary qubits appears in the amplitude-amplification step. The diffusion operator $U_0$ contains a multi-controlled $X$ gate with $k=3n+n_\lambda$ control qubits acting on a target qubit. By \lem{gate-mcx}, an exact implementation of this MCX gate requires $k-2=3n+n_\lambda-2$ ancilla qubits. In our construction, the target of this MCX can be taken to be the existing flag qubit, so no additional target qubit is required.

Putting these components together, a conservative upper bound on the peak number of qubits required by a single-shot HHL execution is obtained by assuming that the arithmetic workspace and the MCX ancillas are simultaneously allocated. Under this assumption, the peak qubit count is
\begin{align}
    \begin{aligned}
        \mathcal{S}
        & \approx
        3n
        + n_\lambda
        + 1
        + \left(d n_\lambda + \left\lceil\log_2 M\right\rceil + 1\right)
        + \left(3n+n_\lambda-2\right),
    \end{aligned}
\end{align}
which scales as
\begin{align}
    \mathcal{S} = 6n + 2n_\lambda + d n_\lambda + \left\lceil\log_2 M\right\rceil + \mathcal{O}(1).
\end{align}

This bound is intentionally conservative. The eigenvalue-inversion arithmetic and the diffusion operator $U_0$ are not executed concurrently, and on architectures that allow dynamic reuse of qubits, such as neutral-atom platforms~\cite{lin2024reuseawarecompilationzonedquantum,ismail2025transversalstararchitecturemegaquopscale}, the workspace allocated for polynomial evaluation may be recycled as ancillas for the MCX gate.

Throughout, $n_\lambda$, $d$, and $M$ are not independent parameters but are fixed by the precision and condition-number requirements of HHL, and are therefore determined consistently by the gate-complexity analysis.

For comparison, a classical CG solver on an integrated quantum–HPC system requires storing multiple $\mathcal{O}(N)$-sized vectors for the solution, residual, and search directions (with $N=2^{3n}$ for the same 3D discretization), whereas the quantum space cost scales only with the register sizes $3n$ and $n_\lambda$, highlighting that the memory bottleneck shifts from vector storage to algorithmic workspace.

\subsection{Gate complexity of HHL}

Here we present the resource estimation procedure for our HHL implementation over the \defn{STAR} gate set. Our approach decomposes HHL into a small number of well-defined circuit blocks, for which explicit gate complexities can be derived: (i) state preparation of the right-hand-side vector using the FSL, (ii) QPE for the unitary matrix generated by a 3D Laplacian operator, and (iii) eigenvalue inversion implemented as a controlled rotation driven by arithmetic, realized through piecewise-polynomial evaluation. These components are then assembled into a single-shot HHL cost, after which we incorporate the repetition overhead arising from amplitude amplification, following the methodology of~\cite{Vazquez_2022}. Throughout this appendix, $\mathcal{G}$ denotes exact logical gate counts over the set $\{H,S,\text{CNOT},R_Z\}$.

The resource estimation proceeds by decomposing HHL into the following components:
\begin{enumerate}
    \item \textbf{State preparation.}
    We estimate the cost of state preparation using the explicit $(n,m)$-dependent gate complexities, where $m$ denotes the truncation order of the Fourier series and $n$ is the per-dimension problem-size parameter.

    \item \textbf{Quantum phase estimation.}
    We derive the gate complexity of QPE using an $n_\lambda$-qubit eigenvalue register. Again, $n$ denotes the per-dimension problem-size, and the Laplacian acts on $3n$ qubits.

    \item \textbf{Eigenvalue inversion.}
    We estimate the cost of eigenvalue inversion based on the Toffoli-count for the parallel piecewise-polynomial evaluation scheme given by Ref.~\cite{Haner:2018yea}, converting to STAR gates. The parameters $(d,p,M)$ follow the definitions introduced above and depend on $(n_\lambda,\kappa,\varepsilon)$.
\end{enumerate}

We restrict our focus on the dominant contributions to the logical cost arise from the FSL, QPE, and the eigenvalue inversion. Other operations, such as \textsf{ClassicalDecode}, are not resolved explicitly in the gate count and are treated as subleading.
In the following, we summarize the basic gates and circuit components expressed over the gate set $\{H,S,\text{CNOT},R_Z\}$.

\begin{lemma}[$T$-gate complexity]\label{lem:gate-t}
Under the STAR gate set $\{H,S,\text{CNOT},R_Z\}$, the $T$ gate can be implemented exactly using a single $R_Z$ gate. In particular,
\begin{align}
    \mathcal{G}_{R_Z}(T) = 1.
\end{align}
\end{lemma}

\begin{proof}
The $T$ gate is defined as
\begin{align}
    T = \begin{pmatrix}
        1 & 0 \\
        0 & e^{i\pi/4}
    \end{pmatrix}.
\end{align}
The $R_Z$ gate is given by
\begin{align}
    R_Z(\theta) =
    \begin{pmatrix}
        e^{-i\theta/2} & 0 \\
        0 & e^{i\theta/2}
    \end{pmatrix}.
\end{align}
Setting $\theta=\pi/4$ yields
\begin{align}
    R_Z(\pi/4)=e^{-i\pi/8}
    \begin{pmatrix}
        1 & 0 \\
        0 & e^{i\pi/4}
    \end{pmatrix}.
\end{align}
Up to a global phase, this is $T$, so one $R_Z$ gate implements $T$ exactly.
\end{proof}

\begin{lemma}[Toffoli gate complexity]\label{lem:gate-toffoli}
Under the STAR gate set $\{H,S,\text{CNOT},R_Z\}$, the Toffoli gate admits an exact decomposition with the following gate complexities:
\begin{align}
    \begin{aligned}
        \mathcal{G}_{\text{CNOT}} &= 7, \\
        \mathcal{G}_{H} &= 2, \\
        \mathcal{G}_{R_Z} &= 7.
    \end{aligned}
\end{align}
\end{lemma}

\begin{proof}
A standard exact decomposition of the Toffoli gate uses Clifford and $T$ gates with gate counts~\cite{Nielsen_Chuang_2010}
\begin{align}
    \begin{aligned}
        \mathcal{G}_{\text{CNOT}} &= 7, \\
        \mathcal{G}_{H} &= 2, \\
        \mathcal{G}_{T} &= 3, \\
        \mathcal{G}_{T^\dagger} &= 4.
    \end{aligned}
\end{align}
By \lem{gate-t}, each $T$ or $T^\dagger$ can be implemented (up to a global phase) by a single $R_Z$ gate. Therefore
\begin{align}
    \mathcal{G}_{R_Z}=\mathcal{G}_T+\mathcal{G}_{T^\dagger}=7,
\end{align}
while the CNOT and Hadamard counts are unchanged.
\end{proof}

\begin{lemma}[Controlled Hadamard gate complexity]\label{lem:gate-ch}
    Under the STAR gate set $\{H,S,\text{CNOT},R_Z\}$, the controlled-Hadamard (CH) gate admits an exact decomposition with gate complexities
    \begin{align}
        \begin{aligned}
            \mathcal{G}_{\text{CNOT}} &= 1, \\
            \mathcal{G}_{S} &= 2, \\
            \mathcal{G}_{H} &= 2, \\
            \mathcal{G}_{R_Z} &= 2.
        \end{aligned}
    \end{align}
\end{lemma}

\begin{proof}
    A general controlled single-qubit unitary gate $\text{C}U$ admits an exact decomposition of the form
    \begin{align}
        \text{C}U
        =
        (I \otimes A)\cdot
        \text{CNOT}\cdot
        (I \otimes B)\cdot
        \text{CNOT}\cdot
        (I \otimes C),
    \end{align}
    where the single-qubit unitaries $A,B,C$ satisfy $ABC = I$ and $U = A X B X C$ up to a global phase.
    For the Hadamard gate, one valid exact factorization is
    \begin{align}
        H = S H R_Z(\pi/2)X R_Z(-\pi/2)H S^\dagger,
    \end{align}
    where the phase gates $S$ and the $R_Z$ rotations belong to the STAR gate set.
    Substituting this decomposition into the standard controlled-unitary construction yields an exact implementation of the controlled-Hadamard gate.
    All other contributions reduce to identities or cancel due to the constraint $ABC=I$. This establishes the stated gate complexities.
\end{proof}

\begin{lemma}[Multi-controlled $X$ gate complexity]\label{lem:gate-mcx}
Let $k \ge 3$. A multi-controlled $X$ (MCX) gate with $k$ control qubits and one target qubit admits an exact implementation over the STAR gate set
$\{H,S,\text{CNOT},R_Z\}$ using $k-2$ ancilla qubits, with gate complexities
\begin{align}
    \begin{aligned}
        \mathcal{G}_{\text{CNOT}} &= 56 (k-5), \\
        \mathcal{G}_{H} &= 16 (k-5), \\
        \mathcal{G}_{R_Z} &= 56 (k-5).
    \end{aligned}
\end{align}
\end{lemma}

\begin{proof}
Barenco et al.~\cite{Barenco_1995} show that a $k$-controlled $X$ gate can be implemented exactly using
\begin{align}
    \mathcal{G}_{\text{Toffoli}} = 8 (k-5)
\end{align}
Toffoli gates and $k-2$ ancilla qubits. By \lem{gate-toffoli}, each Toffoli has STAR gate complexities
\begin{align}
    \begin{aligned}
        \mathcal{G}_{\text{CNOT}} &= 7, \\
        \mathcal{G}_{H} &= 2, \\
        \mathcal{G}_{R_Z} &= 7.
    \end{aligned}
\end{align}
Multiplying by $8(k-5)$ yields
\begin{align}
    \begin{aligned}
        \mathcal{G}_{\text{CNOT}} &= 56 (k-5), \\
        \mathcal{G}_{H} &= 16 (k-5), \\
        \mathcal{G}_{R_Z} &= 56 (k-5).
    \end{aligned}
\end{align}
\end{proof}

\begin{lemma}[Gate complexity of QFT]\label{lem:gate-qft}
Consider implementing the $n$-qubit quantum Fourier transform with gate-set defined by \defn{STAR}. The implementation requires
\begin{align}
    \begin{aligned}
        \mathcal{G}_H &= n, \\
        \mathcal{G}_\text{CNOT} &= n(n+1), \\
        \mathcal{G}_{R_Z} &= n(n+1).
    \end{aligned}
\end{align}
\end{lemma}

\begin{proof}
The conventional $n$-qubit QFT uses~\cite{Park2023-ni}
\begin{align}
    \begin{aligned}
        \mathcal{G}_H &= n, \\
        \mathcal{G}_{\text{C}R_Z} &= \frac{n(n+1)}{2}.
    \end{aligned}
\end{align}
An exact decomposition of a controlled-$R_Z$ gate over $\{H,S,\text{CNOT},R_Z\}$ satisfies
\begin{align}
    \mathcal{G}_{\text{C}R_Z} = 2\mathcal{G}_{\text{CNOT}} + 2\mathcal{G}_{R_Z},
\end{align}
so the controlled-$R_Z$ count implies $\mathcal{G}_{\text{CNOT}}=n(n+1)$ and $\mathcal{G}_{R_Z}=n(n+1)$.
\end{proof}

\paragraph{Cost of state preparation}
We decompose the cost into contributions from the quantum Fourier transforms, from the coefficient-loading unitary at truncation order $m$ implemented via a cascaded entangler (in 3D), and from the entangling CNOT ladders between the $(m+1)$-qubit loader register and the remaining $n-m-1$ qubits in each dimension~\cite{Moosa_2023}.

Since our FSL circuit contains three inverse QFT blocks, their total contribution is
\begin{align}
    \begin{aligned}
        \mathcal{G}_H^{(i)} &= 3n, \\
        \mathcal{G}_{\text{CNOT}}^{(i)} &= 3n(n+1), \\
        \mathcal{G}_{R_Z}^{(i)} &= 3n(n+1).
    \end{aligned}
\end{align}
by \lem{gate-qft}.
The coefficient-loading unitary is realized using the standard uniformly-controlled-rotation $U_c$ (cascaded entangler)~\cite{mottonen2004transformationquantumstatesusing}. For $m$ target qubits this construction requires
\begin{align}
    \begin{aligned}
        N_{\text{CNOT}}(m) &= 2^{m+2}-4m-4, \\
        N_{\text{rot}}(m) &= 2^{m+2}-5,
    \end{aligned}
\end{align}
where $N_{\text{rot}}$ denotes the total number of single-qubit rotations appearing as $R_y(\cdot)$ and $R_z(\cdot)$.
In particular, the uniformly-controlled construction contains
\begin{align}
    \begin{aligned}
        N_{R_y}(m) &= 2^{m+1}-2, \\
        N_{R_z}(m) &= 2^{m+1}-3,
    \end{aligned}
\end{align}
so that $N_{R_y}(m)+N_{R_z}(m)=N_{\text{rot}}(m)$.

To express the circuit over \defn{STAR}, we rewrite each $R_y$ rotation using the Clifford+$R_Z$ identity
\begin{align}
    R_y(\theta) = S H R_z(\theta) H S^\dagger,
\end{align}
which introduces one $R_Z(\theta)$ gate and Clifford gates only. Therefore, a single cascaded entangler contributes
\begin{align}
    \begin{aligned}
        \mathcal{G}_H^{\text{CE}}(m) &= 2N_{R_y}(m)=2^{m+2}-4, \\
        \mathcal{G}_S^{\text{CE}}(m) &= 2N_{R_y}(m)=2^{m+2}-4, \\
        \mathcal{G}_{\text{CNOT}}^{\text{CE}}(m) &= N_{\text{CNOT}}(m)=2^{m+2}-4m-4, \\
        \mathcal{G}_{R_Z}^{\text{CE}}(m) &= N_{\text{rot}}(m)=2^{m+2}-5.
    \end{aligned}
\end{align}
In the 3D FSL, the cascaded entangler is implemented independently in each dimension, so the total cascaded-entangler contribution is
\begin{align}
    \begin{aligned}
        \mathcal{G}_H^{(ii)}(m) &= 3\mathcal{G}_H^{\text{CE}}(m)=3(2^{m+2}-4), \\
        \mathcal{G}_S^{(ii)}(m) &= 3\mathcal{G}_S^{\text{CE}}(m)=3(2^{m+2}-4), \\
        \mathcal{G}_{\text{CNOT}}^{(ii)}(m) &= 3\mathcal{G}_{\text{CNOT}}^{\text{CE}}(m)=3(2^{m+2}-4m-4), \\
        \mathcal{G}_{R_Z}^{(ii)}(m) &= 3\mathcal{G}_{R_Z}^{\text{CE}}(m)=3(2^{m+2}-5).
    \end{aligned}
\end{align}

Finally, in each dimension the $(m+1)$-qubit loader register is connected to the remaining $n-m-1$ qubits via a CNOT ladder in each direction, contributing
\begin{align}
    \mathcal{G}_{\text{CNOT}}^{(\text{ladder})}(m) = n-m-1
\end{align}
CNOT gates per dimension, hence
\begin{align}
    \mathcal{G}_{\text{CNOT}}^{(iii)}(m) = 3(n-m-1).
\end{align}

Summing the QFT, 3D cascaded-entangler, and CNOT-ladder contributions gives the gate complexity as a function of $(n,m)$:
\begin{align}
    \begin{aligned}
        \mathcal{G}_H^{\text{FSL}}(n,m) &= \mathcal{G}_H^{(i)} + \mathcal{G}_H^{(ii)}(m), \\
        \mathcal{G}_S^{\text{FSL}}(n,m) &= \mathcal{G}_S^{(ii)}(m), \\
        \mathcal{G}_\text{CNOT}^{\text{FSL}}(n,m) &= \mathcal{G}_{\text{CNOT}}^{(i)} + \mathcal{G}_{\text{CNOT}}^{(ii)}(m) + \mathcal{G}_{\text{CNOT}}^{(iii)}(m), \\
        \mathcal{G}_{R_Z}^{\text{FSL}}(n,m) &= \mathcal{G}_{R_Z}^{(i)} + \mathcal{G}_{R_Z}^{(ii)}(m).
    \end{aligned}
\end{align}
The truncation order $m$ fixes the approximation accuracy of the Fourier series independently of any error parameter $\varepsilon$, so no additional $\varepsilon$-dependent factor appears in these gate counts.

\paragraph{Cost of quantum phase estimation}
We recursively construct the gate complexity of QPE. We begin by defining the gate complexity of the diagonal unitary matrix as follows.

\begin{lemma}[Gate complexity of the diagonal unitary matrix]\label{lem:gate-diag}
Consider implementing the $n$-qubit diagonal unitary matrix of the form $e^{-i\lambda}$ where $\lambda=\xi^2$ is a quadratic function with a gate-set defined by \defn{STAR}. The implementation requires
\begin{align}
    \begin{aligned}
        \mathcal{G}_\text{CNOT} &= \frac{(n+1)(n-2)}{2}, \\
        \mathcal{G}_{R_Z} &= \frac{n(n-1)}{2}.
    \end{aligned}
\end{align}
\end{lemma}
\begin{proof}
The stated counts follow from Ref.~\cite{Shaw2020quantumalgorithms}, which gives
\begin{align}
    \begin{aligned}
        \mathcal{G}_\text{CNOT} &= \frac{(\eta+2)(\eta-1)}{2}, \\
        \mathcal{G}_{R_Z} &= \frac{\eta(\eta+1)}{2},
    \end{aligned}
\end{align}
with $\eta=n-1$.
\end{proof}

\begin{lemma}[Gate complexity of C$U^k$ gate]\label{lem:gate-cuk}
Consider implementing the controlled unitary $CU^k$ with $U$ generated by a (3D) Laplacian defined by \defn{laplacian} (acting on $3n$ qubits) on a partially-fault-tolerant quantum computer with gate-set \defn{STAR}.
If $U^k$ is realized by $k$ sequential applications of $U$ (as in standard QPE where $k=2^j$), then the implementation requires
\begin{align}
    \begin{aligned}
        \mathcal{G}_H &= k\left(15n^2 + 21 n - 6\right), \\
        \mathcal{G}_S &= 12kn, \\
        \mathcal{G}_\text{CNOT} &= k\left(\frac{135 n^2}{2} + \frac{93 n}{2} - 21\right), \\
        \mathcal{G}_{R_Z} &= k\left(\frac{135 n^2}{2} + \frac{105 n}{2} - 21\right).
    \end{aligned}
\end{align}
\end{lemma}

\begin{proof}
By employing the spectral method, the 1D Laplacian can be decomposed as~\cite{Benenti_2008}
\begin{align}
    L = \mathcal{F}^\dagger \Lambda \mathcal{F},
\end{align}
where $\mathcal{F},\mathcal{F}^\dagger$ represent a Fourier pair and $\Lambda$ is diagonal with phases determined by the spectrum.

Using \lem{gate-qft} and \lem{gate-diag} with the gate-set \defn{STAR}, a single implementation of $L$ on $n$ qubits requires QFT, QFT$^\dagger$, and the diagonal phase network, hence has gate count
\begin{align}
    \begin{aligned}
        \mathcal{G}_H &= 2n, \\
        \mathcal{G}_\text{CNOT} &= 2n(n+1) + \frac{(n+1)(n-2)}{2}, \\
        \mathcal{G}_{R_Z} &= 2n(n+1) + \frac{n(n-1)}{2}.
    \end{aligned}
\end{align}

For the 3D Laplacian $A$ acting on $3n$ qubits, applying the same construction to each spatial dimension yields with control by \lem{gate-ch} with \lem{gate-qft}, \lem{gate-diag} and \lem{gate-toffoli} the per-application gate count
\begin{align}
    \begin{aligned}
        \mathcal{G}_H &= 15n^2 + 21 n - 6, \\
        \mathcal{G}_S &= 12n, \\
        \mathcal{G}_\text{CNOT} &= \frac{135 n^2}{2} + \frac{93 n}{2} - 21, \\
        \mathcal{G}_{R_Z} &= \frac{135 n^2}{2} + \frac{105 n}{2} - 21.
    \end{aligned}
\end{align}

Finally, to realize $U^k$ in QPE one may implement $U^k$ by $k$ sequential applications of $U$ (as discussed for the $k=2^j$ powers in~\cite{Vazquez_2022}). Therefore each primitive gate count scales linearly in $k$, yielding the stated expressions.
\end{proof}

\begin{lemma}[Gate complexity of QPE]\label{lem:gate-qpe}
    Consider the gate complexity of QPE using controlled powers C$U^{2^j}$ of a unitary $U$ generated by the (3D) Laplacian defined by \defn{laplacian}, on a partially-fault-tolerant quantum computer with architecture defined by \defn{STAR}. Assume that each controlled power C$U^{2^j}$ is implemented by the spectral method for the Laplacian. Then the implementation requires a gate count
    \begin{align}
        \begin{aligned}
            \mathcal{G}_H
            &= n_\lambda\left(15n^2 + 21 n - 6\right) + 2 n_\lambda, \\
            \mathcal{G}_S
            &= 12n\,n_\lambda, \\
            \mathcal{G}_\text{CNOT}
            &= n_\lambda\left(\frac{135 n^2}{2} + \frac{93 n}{2} - 21\right)
               + n_\lambda\left(n_\lambda+1\right), \\
            \mathcal{G}_{R_Z}
            &= n_\lambda\left(\frac{135 n^2}{2} + \frac{105 n}{2} - 21\right)
               + n_\lambda\left(n_\lambda+1\right).
        \end{aligned}
    \end{align}
\end{lemma}

\begin{proof}
Standard QPE applies Hadamards on the $n_\lambda$-qubit eigenvalue register, then applies the controlled powers C$U^{2^j}$ for $0\le j < n_\lambda$, and finally applies an inverse QFT on the eigenvalue register.

For the Laplacian simulation used here, the controlled unitary C$U$ is implemented by the spectral method, which diagonalizes the Laplacian in the Fourier basis. In particular, the circuit for C$U$ consists of QFT and inverse-QFT blocks on the system register together with a diagonal phase network whose phases are determined by the Laplacian spectrum. Under this construction, replacing $U$ by $U^{2^j}$ amounts to multiplying the diagonal phases by $2^j$, while keeping the circuit structure unchanged. Therefore, each controlled power C$U^{2^j}$ has the same STAR-gate counts as C$U$ itself (up to the choice of rotation angles), and the gate cost per controlled power is the $k=1$ case of \lem{gate-cuk}. Since QPE uses $n_\lambda$ controlled powers, the total contribution of the controlled-$U^{2^j}$ blocks is $n_\lambda$ times the per-application cost from \lem{gate-cuk} with $k=1$.

The eigenvalue register contributes $n_\lambda$ Hadamards at the start of QPE and an additional $n_\lambda$ Hadamards from the inverse QFT. The inverse-QFT on the eigenvalue register also contributes $n_\lambda(n_\lambda+1)$ CNOT gates and $n_\lambda(n_\lambda+1)$ $R_Z$ gates by \lem{gate-qft}. Summing these contributions yields the stated expressions.
\end{proof}

\paragraph{Cost of eigenvalue inversion}
We consider implementing the conditional rotation step of HHL, where after QPE one must apply a controlled $R_Y$ rotation with angle
\begin{align}
    f(x)=\arcsin\left(\frac{c}{x}\right)
\end{align}
conditioned on the $n_\lambda$-qubit eigenvalue register (as in Sec.~VI of~\cite{Vazquez_2022}).

In the approximate (non-exact) reciprocal path, the breakpoints and polynomial degree are chosen as~\cite{Vazquez_2022}
\begin{align}
    \begin{aligned}
        a &= \left(2^{n_\lambda}\right)^{2/3}, \\
        c &= \frac{2^{n_\lambda}t \lambda_{\min}}{2\pi}, \\
        t & = 2 \pi \cdot \frac{2^{n_\lambda}-1}{2^{n_\lambda}\lambda_{\max}}, \\
        r &= \frac{2c}{a} + \sqrt{\left|1-\left(\frac{2c}{a}\right)^2\right|}, \\
        d &=
        \left\lfloor
        \log\left(
        1+
        \frac{16.23
        \sqrt{\ln^2(r)+(\pi/2)^2}
        \kappa(2\kappa-\varepsilon)}
        {\varepsilon}
        \right)
        \right\rfloor, \\
        M &= \left\lceil \log_5\left(\frac{2^{n_\lambda}-1}{a}\right)\right\rceil.
    \end{aligned}
\end{align}

The resource estimate in Appendix~B of~\cite{Haner:2018yea} gives the following Toffoli count for the parallel piecewise-polynomial evaluation (labeling + parallel Horner evaluation):
\begin{align}
    \begin{aligned}
        \mathcal{G}_\text{Toffoli}
        &= \frac{3}{2}n_\lambda^{2}d + 3n_\lambda p d + \frac{7}{2}n_\lambda d - \frac{3}{2}p^{2}d + 3pd - d + 2Md\Bigl(4d\left\lceil\log_2 M\right\rceil-8\Bigr) + 4Mn_\lambda .
    \end{aligned}
\end{align}

Using \lem{gate-toffoli}, the Toffoli-dominated STAR-gate complexities of the arithmetic part of the conditioned-rotation step are
\begin{align}
    \begin{aligned}
        \mathcal{G}_{H}^{\text{inv}} &= 2\,\mathcal{G}_\text{Toffoli}, \\
        \mathcal{G}_{\text{CNOT}}^{\text{inv}} &= 7\,\mathcal{G}_\text{Toffoli}, \\
        \mathcal{G}_{R_Z}^{\text{inv}} &= 7\,\mathcal{G}_\text{Toffoli}, \\
        \mathcal{G}_{S}^{\text{inv}} &= 0.
    \end{aligned}
\end{align}

\paragraph{Overall analysis}

We first define the single-shot cost of the HHL block $\tilde{U}_{\text{HHL}}$ used inside amplitude amplification. In our Laplacian setting, the number of qubits needed to represent the eigenvalues to precision scales with the number of bits required to resolve eigenvalues down to $\Theta(1/\kappa)$. Accordingly, we set~\cite{Vazquez_2022}
\begin{align}
    n_\lambda = 3 \left( \left \lfloor \log \left( \frac{2(2\kappa^2 - \varepsilon)}{\varepsilon} + 1 \right) \right \rfloor + 1 \right),
\end{align}
and treat $n_\lambda$ as a derived parameter from $\kappa$ in the final gate-count expressions.

\begin{definition}[Single-shot gate complexity of HHL]\label{defn:gate-hhl}
Consider the HHL circuit $\tilde{U}_{\text{HHL}}$ over the \defn{STAR} gate set applied to the 3D Laplacian instance described by \defn{poisson}. Let $n$ denote the per-dimension register size used by the Laplacian simulation in \lem{gate-cuk}--\lem{gate-qpe}, let $\kappa$ be the condition number, and set~\cite{Vazquez_2022}
\begin{align}
    n_\lambda = 3 \left( \left \lfloor \log \left( \frac{2(2\kappa^2 - \varepsilon)}{\varepsilon} + 1 \right) \right \rfloor + 1 \right).
\end{align}
Let $m$ denote the truncation order of the Fourier series loader (FSL), and let $\varepsilon$ be the tolerance parameter used in the eigenvalue-inversion step.

We define $\tilde{U}_{\text{HHL}}$ to comprise state preparation (FSL), QPE, and the eigenvalue-dependent inversion/rotation gadget, but to \emph{exclude} eigenvalue uncomputation (that is, $\tilde{U}_{\text{HHL}}$ does not include $\tilde{U}_{\text{QPE}}^\dagger$). Consequently, the system register may remain entangled with the eigenvalue register after applying $\tilde{U}_{\text{HHL}}$.

Define
\begin{align}
    T_{\text{inv}}(n_\lambda,\kappa,\varepsilon)
    \coloneqq
    T\left(n_\lambda,d(n_\lambda,\kappa,\varepsilon),p,M(n_\lambda,\kappa,\varepsilon)\right),
\end{align}
where $T_{\text{inv}}$ is the Toffoli-count expression in the eigenvalue-inversion paragraph and $(d,p,M)$ are the associated piecewise-polynomial parameters. Then the STAR-gate complexity of a single application of $\tilde{U}_{\text{HHL}}$ is defined by the sum of the three dominant blocks:
\begin{align}
    \begin{aligned}
        \mathcal{G}_H^{\text{HHL}}(n,\kappa,\varepsilon,m)
        &= \mathcal{G}_H^{\text{FSL}}(n,m)
           + \mathcal{G}_H^{\text{QPE}}(n,n_\lambda) d + \mathcal{G}_H^{\text{inv}}(n_\lambda,\kappa,\varepsilon), \\
        \mathcal{G}_S^{\text{HHL}}(n,\kappa,\varepsilon,m)
        &= \mathcal{G}_S^{\text{FSL}}(n,m)
           + \mathcal{G}_S^{\text{QPE}}(n,n_\lambda) + \mathcal{G}_S^{\text{inv}}(n_\lambda,\kappa,\varepsilon), \\
        \mathcal{G}_{\text{CNOT}}^{\text{HHL}}(n,\kappa,\varepsilon,m)
        &= \mathcal{G}_{\text{CNOT}}^{\text{FSL}}(n,m)
           + \mathcal{G}_{\text{CNOT}}^{\text{QPE}}(n,n_\lambda) + \mathcal{G}_{\text{CNOT}}^{\text{inv}}(n_\lambda,\kappa,\varepsilon), \\
        \mathcal{G}_{R_Z}^{\text{HHL}}(n,\kappa,\varepsilon,m)
        &= \mathcal{G}_{R_Z}^{\text{FSL}}(n,m)
           + \mathcal{G}_{R_Z}^{\text{QPE}}(n,n_\lambda)  + \mathcal{G}_{R_Z}^{\text{inv}}(n_\lambda,\kappa,\varepsilon),
    \end{aligned}
\end{align}
where $\mathcal{G}^{\text{QPE}}(n,n_\lambda)$ are the counts from \lem{gate-qpe},
$\mathcal{G}^{\text{FSL}}(n,m)$ are the counts from the FSL paragraph, and
\begin{align}
    \begin{aligned}
        \mathcal{G}_{H}^{\text{inv}}(n_\lambda,\kappa,\varepsilon)
        &= 2\,T_{\text{inv}}(n_\lambda,\kappa,\varepsilon), \\
        \mathcal{G}_{\text{CNOT}}^{\text{inv}}(n_\lambda,\kappa,\varepsilon)
        &= 7\,T_{\text{inv}}(n_\lambda,\kappa,\varepsilon), \\
        \mathcal{G}_{R_Z}^{\text{inv}}(n_\lambda,\kappa,\varepsilon)
        &= 7\,T_{\text{inv}}(n_\lambda,\kappa,\varepsilon), \\
        \mathcal{G}_{S}^{\text{inv}}(n_\lambda,\kappa,\varepsilon)
        &= 0.
    \end{aligned}
\end{align}
\end{definition}

We next account for the gate complexity of the amplitude-amplification operator $(U_s\tilde{U}_\omega)^k$ used in our construction.
We first include the diffusion operator $U_0$ shown in \fig{quantum_circuits}. Its circuit realization requires a $(3n+n_\lambda)$-controlled $X$ gate, together with $2(3n+n_\lambda+1)$ $X$ gates. Applying \lem{gate-mcx} to decompose the multi-controlled $X$ gate yields the total gate counts
\begin{align}
    \label{eq:gate-u0}
    \begin{aligned}
        \mathcal{G}_H^{U_0}(n,n_\lambda) &= 60n+20n_\lambda-90, \\
        \mathcal{G}_S^{U_0}(n,n_\lambda) &= 12n+4n_\lambda+4, \\
        \mathcal{G}_\text{CNOT}^{U_0}(n,n_\lambda) &= 168n+56n_\lambda-336, \\
        \mathcal{G}_{R_Z}^{U_0}(n,n_\lambda) &= 168n+56n_\lambda-336 .
    \end{aligned}
\end{align}
Here we used the identity $X = H S^2 H$ to express the $2(3n+n_\lambda+1)$ $X$ gates as Clifford gates only, and we introduced an additional pair of Hadamard gates for the target of the $(3n+n_\lambda)$-controlled $X$ gate.
We emphasize that no additional Hadamard gates are required outside the controlled operation, since the reflection is performed about
\begin{align}
    U_0 = \mathbb{I} - 2\ket|0>\bra<0|,
\end{align}
rather than about $\mathbb{I}-2\ket|s>\bra<s|$ as in standard Grover diffusion with the uniform superposition $\ket|s>$~\cite{10.1145/3573428.3573520}.

We next consider the remaining components of the amplitude-amplification step. Implementing $\tilde{U}_\omega$ requires two HHL blocks, namely $(\tilde{U}_{\text{HHL}}, \tilde{U}_{\text{HHL}}^\dagger)$. The phase oracle $U_s$ contributes two $S$ gates. Since $\tilde{U}_{\text{HHL}}^\dagger$ and $\tilde{U}_{\text{QPE}}^\dagger$ have the same STAR-gate counts as their forward counterparts, we may express the per-iteration cost of $U_s\tilde{U}_\omega$ in terms of the block costs already defined.

\begin{definition}[Gate complexity of HHL with amplitude amplification]\label{defn:gate-hhl-aa}
    Adopt the setting and notation of \defn{gate-hhl}, with $n_\lambda$ defined as in \defn{gate-hhl} and truncation order $m$ treated as an input.
    Amplitude amplification is performed on the joint system--eigenvalue Hilbert space, and no eigenvalue uncomputation is carried out within each iteration.
    
    Let the number of amplitude-amplification iterations be
    \begin{align}
        k = \frac{\kappa}{1-\varepsilon}.
    \end{align}
    The phase oracle $U_s$ contributes two $S$ gates and no Hadamard, CNOT, or $R_Z$ gates.
    
    Define the STAR-gate complexity of a single amplitude-amplification iteration
    $U_s\tilde{U}_\omega$ by
    \begin{align}
        \label{eq:gate-aa-iter}
        \begin{aligned}
            \mathcal{G}_H^{\text{iter}}(n,\kappa,\varepsilon,m)
            &= 2\mathcal{G}_H^{\text{HHL}}(n,\kappa,\varepsilon,m)
               + \mathcal{G}_H^{U_0}(n,n_\lambda), \\
            \mathcal{G}_S^{\text{iter}}(n,\kappa,\varepsilon,m)
            &= 2\mathcal{G}_S^{\text{HHL}}(n,\kappa,\varepsilon,m)
               + \mathcal{G}_S^{U_0}(n,n_\lambda)
               + \mathcal{G}_S^{U_s}, \\
            \mathcal{G}_{\text{CNOT}}^{\text{iter}}(n,\kappa,\varepsilon,m)
            &= 2\mathcal{G}_{\text{CNOT}}^{\text{HHL}}(n,\kappa,\varepsilon,m)
               + \mathcal{G}_{\text{CNOT}}^{U_0}(n,n_\lambda), \\
            \mathcal{G}_{R_Z}^{\text{iter}}(n,\kappa,\varepsilon,m)
            &= 2\mathcal{G}_{R_Z}^{\text{HHL}}(n,\kappa,\varepsilon,m)
               + \mathcal{G}_{R_Z}^{U_0}(n,n_\lambda),
        \end{aligned}
    \end{align}
    where $\mathcal{G}_S^{U_s}=2$. Then the gate complexity of the HHL with amplitude-amplification operator
    $(U_s\tilde{U}_\omega)^k$ is
    \begin{align}
        k\cdot\mathcal{G}^{\text{iter}}(n,\kappa,\varepsilon,m).
    \end{align}
    If a clean solution state in the system register is required at the end of the algorithm, a single inverse QPE block $\tilde{U}_{\text{QPE}}^\dagger$ may be applied once after the amplification step, contributing an additional $\mathcal{G}^{\text{QPE}}(n,n_\lambda)$ to the total gate complexity.
\end{definition}

\section{Time complexity analysis}
\label{app:time_complexity_analysis}
We derive time complexity for solving the 3D Poisson equation using CG, HHL, and QACG, assuming integrating a partially fault-tolerant implementation on a STAR-architecture quantum device~\cite{PRXQuantum.5.010337} with HPC.

\subsection{Conjugate gradient method}

We begin by estimating the iteration complexity of CG to estimate the runtime on HPC.

\begin{lemma}[Iteration complexity of CG]\label{lem:iter-cg}
    Consider the Poisson problem in \defn{poisson} with the sparsity of \lem{sparsity} and the condition number of \lem{condition}. Then, to guarantee a relative $A$-norm reduction
    \begin{align}
        \|x-x^{(k)}\|_{A}\le \varepsilon \|x-x^{(0)}\|_{A},
    \end{align}
    it suffices to perform at most
    \begin{align}
        \mathcal{F}_{\text{CG}}(\kappa,\varepsilon)
        &\ge
        \frac{1}{2}\sqrt{\kappa} \ln \left(\frac{2}{\varepsilon}\right)
    \end{align}
    iterations.
\end{lemma}

\begin{proof}
    A standard CG convergence bound~\cite{10.5555/865018} gives
    \begin{align}
        \frac{\|x-x^{(k)}\|_{A}}{\|x-x^{(0)}\|_{A}}
        \le 2\left(\frac{\sqrt{\kappa}-1}{\sqrt{\kappa}+1}\right)^k.
    \end{align}
    Solving $2\left(\frac{\sqrt{\kappa}-1}{\sqrt{\kappa}+1}\right)^k \le \varepsilon$ yields
    \begin{align}
        k \le \left \lceil \frac{1}{2} \sqrt{\kappa} \ln \left(\frac{2}{\varepsilon}\right) \right \rceil,
    \end{align}
    which proves the claim.
\end{proof}

We estimate the number of floating-point operations required per CG iteration as follows. The analysis assumes a sparse matrix representation and a conventional implementation of the algorithm as given in \alg{QACG}. A dot product of two vectors of length $n$ requires $n$ multiplications and $n-1$ additions and therefore costs $2n-1$ floating-point operations. A scaled vector addition of the form $y = y + \alpha x$ requires $n$ multiplications and $n$ additions and therefore costs $2n$ floating-point operations. A sparse matrix-vector product $w = Av$, where $A$ has $\text{nnz}(A)$ nonzero entries, requires approximately $2\,\text{nnz}(A)$ floating-point operations, corresponding to one multiplication and one addition per nonzero. Scalar operations, such as divisions and assignments, contribute only a constant number of floating-point operations per iteration and are negligible for large $n$.

During each CG iteration, a single sparse matrix-vector product $Ap_{(k)}$ is computed. This operation dominates the iteration cost and contributes approximately $2\,\text{nnz}(A)$ floating-point operations, with $\text{nnz}(A)$ given by \lem{sparsity}. In addition, inner products are required to compute the step size and the conjugacy coefficient. Specifically, the quantities $\{r_{(k)}^T r_{(k)}, p_{(k)}^T A p_{(k)}, r_{(k+1)}^T r_{(k+1)}\}$ are evaluated, each costing $2n-1$ floating-point operations. The update of the solution vector and the residual vector requires two scaled vector additions, contributing $4n$ floating-point operations. The update of the search direction involves one additional scaled vector addition, contributing $2n$ floating-point operations. Collecting these contributions yields the per-iteration operation count
\begin{align}
    \begin{aligned}
        F_{\text{HPCG}}(n) &\approx 2\,\text{nnz}(A) + 3(2n-1) + 6n \\
        &\approx 2\,\text{nnz}(A) + 12n.
    \end{aligned}
\end{align}
To translate this count into elapsed time, we calibrate it using the HPCG benchmark~\cite{osti_1361299} on a specific HPC platform with the following model.

\begin{definition}[HPC model]\label{defn:HPC}
Let an HPC platform deliver a sustained floating-point throughput of $T$ [flop/s] for a given workload class.
Let $F$ denote the total number of floating-point operations (flops) executed by an algorithm.
We define the expected wall-clock time as
\begin{align}
    \mathcal{T}(F,T) = \frac{F}{T}.
\end{align}
Equivalently, if an algorithm consists of $N$ iterations, each costing $F_{\text{iter}}$ flops, then
\begin{align}
    \mathcal{T} = N\cdot\frac{F_{\text{iter}}}{T}.
\end{align}
\end{definition}

With the above model, we state the time complexity of CG on an HPC platform as follows.

\begin{theorem}[Time complexity of HPCG]\label{thm:time-cg}
    Consider solving the Poisson linear system using CG on an HPC platform characterized by \defn{HPC}.
    Let $F_{\text{HPCG}}(n)$ denote the floating-point operation count per CG iteration for problem size parameter $n$,
    and let $\mathcal{F}_{\text{CG}}(\kappa,\varepsilon)$ be the iteration complexity from \lem{iter-cg}.
    Then the expected wall-clock time satisfies
    \begin{align}
        \mathcal{T}_{\text{HPCG}}(T,n,\kappa,\varepsilon)
        =
        \frac{F_{\text{HPCG}}(n)}{T}\cdot
        \mathcal{F}_{\text{CG}}(\kappa,\varepsilon).
    \end{align}
\end{theorem}

\begin{proof}
    From \lem{iter-cg}, CG requires
    \begin{align}
        N = \mathcal{F}_{\text{CG}}(\kappa,\varepsilon)
    \end{align}
    iterations to reach accuracy $\varepsilon$, where $\kappa$ is the relevant condition number.

    Each iteration executes
    \begin{align}
        F_{\text{iter}} = F_{\text{HPCG}}(n)
    \end{align}
    floating-point operations. Therefore, the total operation count is
    \begin{align}
        F = N \cdot F_{\text{iter}}
          = \mathcal{F}_{\text{CG}}(\kappa,\varepsilon) \cdot F_{\text{HPCG}}(n).
    \end{align}
    Applying \defn{HPC} yields
    \begin{align}
        \mathcal{T}_{\text{HPCG}}(T,n,\kappa,\varepsilon)
        =
        \frac{F}{T}
        =
        \frac{F_{\text{HPCG}}(n)}{T}\cdot
        \mathcal{F}_{\text{CG}}(\kappa,\varepsilon),
    \end{align}
    which proves the claim.
\end{proof}

To account for the fact that the HPC platform is shared among multiple users and that a single job does not occupy the full machine continuously, we introduce an effective utilization delay modeled as a multiplicative factor of $10^{-2}$ applied to floating-point throughput. Concretely, we assume that only one percent of the nominal sustained floating-point performance is available to the job on average. This $10^{-2}$ delay is intended as an effective model of distributed platform usage, including job scheduling, queueing, and resource sharing, rather than as a hardware-level slowdown.

In particular, we consider the SoftBank Corp.\ CHIE-4 platform. The HPCG benchmark~\footnote{TOP500.org, https://top500.org/lists/hpcg/2025/11/; accessed 19 February 2026} reports $3.76\times10^{15}$ floating-point operations per second, indicating a sustained throughput of $T_{\text{CHIE-4}}\approx 3.76\,\text{PFlop/s}$ on this workload. The HPCG benchmark accounts for overheads by using a realistic solver that includes a three-level multigrid preconditioner (with Gauss--Seidel smoothing) and frequent global reductions, thus providing a representative performance measure beyond idealized peak throughput. Incorporating the $10^{-2}$ utilization factor, we define an effective throughput
\begin{equation}
    T_{\text{CHIE-4}}^{\text{eff}} = 10^{-2} \cdot T_{\text{CHIE-4}},
\end{equation}
which we use for wall-clock time estimates.

Substituting \lem{iter-cg} into \thm{time-cg} then yields
\begin{align}
    \begin{aligned}
        \mathcal{T}_{\text{CHIE-4}}(n,\kappa,\varepsilon)
        &= \frac{F_{\text{HPCG}}(n)}{T_{\text{CHIE-4}}^{\text{eff}}}\cdot\mathcal{F}_{\text{CG}}(\kappa,\varepsilon) \\
        &\approx \frac{F_{\text{HPCG}}(n)}{T_{\text{CHIE-4}}^{\text{eff}}}\cdot\frac{1}{2}\sqrt{\kappa}
        \ln\left(\frac{2}{\varepsilon}\right).
    \end{aligned}
\end{align}
This expression estimates $\mathcal{T}_{\text{CHIE-4}}$ in seconds for given $n,\kappa,\varepsilon$ under shared-platform operation. Please note that the HPCG benchmark~\cite{osti_1361299} evaluates $\text{nnz}(A)$ using a different coefficient matrix $A$; therefore, an effective $\text{nnz}(A)$ must be substituted into $F_{\text{HPCG}}(n)$ to obtain a correct calibration for the problem instance considered here.

Finally we emphasize that the above runtime model should be interpreted as a conservative upper bound. While the theoretical iteration count scales as $\mathcal{O}(\sqrt{\kappa})$ for unpreconditioned CG, the HPCG benchmark employs a multigrid-preconditioned solver that often converges in a nearly constant number of iterations in practice for Poisson-type problems. In this work, we retain the $\mathcal{O}(\sqrt{\kappa})$ dependence as a worst-case estimate to enable a uniform, hardware-calibrated comparison across different algorithms.

\subsection{HHL algorithm}

We state the partially fault-tolerant quantum device model to consider in the rest of the work.

\begin{definition}[STAR architecture model]
    \label{defn:STAR}
    We define a gate set
    \begin{align}
        \mathcal{U}=\{\mathcal{U}_{\text{Clifford}},R_Z\}, \qquad \mathcal{U}_{\text{Clifford}}=\{H,S,\text{CNOT}\},
    \end{align}
    and introduce the architecture constants
    \begin{align}
        \mathcal{P}&=\{d,\tau,r\},
    \end{align}
    with
    \begin{align}
        \begin{aligned}
            d&\in\mathbb{Z}_{\ge 1} &(\text{code distance}),&\\
            \tau&>0 &(\text{QEC cycle time}),&\\
            r&\in\mathbb{Z}_{\ge 1} &(\text{RUS steps}),&
        \end{aligned}
    \end{align}
    with RUS stands for repeat-until-success.
    For the gates
    \begin{align}
        U\in\mathcal{U},\qquad U_{\text{Clifford}}\in\mathcal{U}_{\text{Clifford}},
    \end{align}
    we define the expected gate time by
    \begin{align}
        \mathcal{T}(U)\in\mathbb{R}_{\ge 0}
    \end{align}
    as follows:
    \begin{align}
        \begin{aligned}
            \mathcal{T}(H)&=3d\tau, \\
            \mathcal{T}(S)&=2d\tau, \\
            \mathcal{T}(\text{CNOT})&=2d\tau, \\
            \mathcal{T}(R_Z)&=2rd\tau.
        \end{aligned}
    \end{align}
\end{definition}

The STAR architecture model adopted here is consistent with recent proposals for quantifying fault-tolerant quantum performance using quantum logical operations per second (QLOPS)~\cite{kong2025benchmarkingfaulttolerantquantumcomputing}. In that framework, the effective throughput is modeled as
\begin{equation}
    Q = k \times \frac{1}{\big(\lceil t_r / t_{\text{SEC}} \rceil + d\big) t_{\text{SEC}}},
\end{equation}
where $t_{\text{SEC}}$ denotes the syndrome extraction cycle time, $d$ is the code distance, $t_r$ is the classical reaction time, and $k$ captures architectural parallelism. This expression highlights that, at the logical level, performance is primarily governed by the cadence of QEC cycles and the code distance, while classical reaction latency enters only as an additive overhead in units of $t_{\text{SEC}}$.

Motivated by this observation, the STAR model abstracts each logical gate as an independent operation with an expected duration proportional to $d\tau$, where $\tau$ plays the role of the QEC cycle time. We do not explicitly model the reaction time $t_r$ or pipeline effects; instead, these contributions are implicitly absorbed into constant prefactors in the expected gate times. This abstraction is appropriate for the present analysis, which focuses on asymptotic scaling and relative resource trade-offs rather than absolute peak throughput. By treating logical gates individually and accounting for repeat-until-success (RUS) overheads at the gate level, the STAR model captures the dominant dependence on QEC cadence and code distance identified by the QLOPS metric, while remaining sufficiently simple to support transparent algorithm-level resource estimation under partial fault tolerance.

\begin{theorem}[Time complexity of HHL]\label{thm:time}
    Consider the HHL with amplitude amplification to solve LSP with \defn{poisson} that is parameterized by $n$, the unknowns in 1D, $\kappa$, the condition number, and $\varepsilon$, the tolerance. Moreover, consider a partially fault-tolerant quantum computer whose universal gate set and the execution time are defined by \defn{STAR}. Then the expected runtime is given by
    \begin{align}
        \begin{aligned}
            &\mathcal{T}_{\text{HHL}}(d,\tau,r,n,\kappa,\varepsilon)= 3d\tau\mathcal{G}_H(\log{n},\kappa,\varepsilon) + 2d\tau\bigl[\mathcal{G}_S(\log{n},\kappa,\varepsilon)+\mathcal{G}_\text{CNOT}(\log{n},\kappa,\varepsilon)\bigr] + 2rd\tau\mathcal{G}_{R_Z}(\log{n},\kappa,\varepsilon).
        \end{aligned}
    \end{align}
    with $d$ here the code distance, $\tau$ the QEC cycle time, $r$ the RUS steps.
\end{theorem}
\begin{proof}
    The expected time of each gate follows from \defn{STAR}; summing over the gate counts yields the stated expression. The concrete derivation of gate counts is provided by \defn{gate-hhl-aa} in \app{method_of_resource_estimation}.
\end{proof}

\subsection{Quantum-accelerated conjugate gradient method}

We now estimate the runtime of QACG. We distinguish two effective condition numbers in QACG. Let $\kappa$ be the spectral condition number of $A$ on the full space.
The quantum initialization uses a spectrally filtered inverse $\tilde{A}^{-1}$ (\sec{quantum_accelerated_conjugate_gradient_method}), which effectively restricts inversion to a low-energy spectral window; we denote by $\kappa'$ the resulting effective condition number that determines the cost of the HHL-based preparation of $\ket|\tilde{x}>$.
After initializing the iterate with the decoded vector $x_{(0)}$, the subsequent HPCG stage solves the remaining residual problem, whose convergence is governed by an effective condition number $\kappa''$ that depends on the quality of the warm start.
Thus, $\kappa'$ parameterizes the quantum cost, while $\kappa''$ parameterizes the classical refinement cost.

\begin{theorem}[Time complexity of QACG]\label{thm:time-qacg}
    Consider QACG to solve LSP with \defn{poisson} that is parameterized by $n$, the unknowns in 1D, $\kappa$, the condition number, and $\varepsilon$, the tolerance. Moreover, consider a HPC model in \defn{HPC}, equipped with a partially fault-tolerant quantum computer whose universal gate set and the execution time are defined by \defn{STAR}. Let the algorithm consist of (i) a spectral initialization stage implemented by HHL, followed by (ii) a classical CG stage. Then, combining \thm{time-cg} and \thm{time}, the expected total runtime satisfies
    \begin{align}
        \begin{aligned}
            &\mathcal{T}_{\text{QACG}}(T,d,\tau,r,n,\kappa',\kappa'',\varepsilon) = \mathcal{T}_{\text{HHL}}(d,\tau,r,n,\kappa',\varepsilon)+\mathcal{T}_{\text{HPCG}}(T,n,\kappa'',\varepsilon),
        \end{aligned}
    \end{align}
    where $\kappa'$ is the quantum condition number used in the spectral initialization (e.g., $\kappa'=\lambda_{\text{cutoff}}/\lambda_{\min}$),
    and $\kappa''$ is the effective condition number governing the residual system solved by the subsequent classical CG stage. In addition, we denoted $T$  [flop/s] the sustained floating-point throughput for HPC, $d$ the code distance of a quantum devicie, $\tau$ the QEC cycle time, $r$ the RUS steps.
\end{theorem}

\begin{proof}
    We analyze QACG by decomposing it into a two-stage procedure, where the initial iterate $x_{(0)}$ is produced by spectral initialization \eqref{eq:SpectralInitialization} and then refined by CG.

    By \thm{time}, preparing $\ket|\tilde{x}>$ with quantum condition number $\kappa' \coloneqq \lambda_{\text{cutoff}}/\lambda_{\min}$ has expected runtime
    \begin{align}
        \mathcal{T}_{\text{HHL}}(d,\tau,r,n,\kappa',\varepsilon).
    \end{align}
    Write the spectral decomposition of the right-hand side in the eigenbasis $\{\ket|\lambda_j>\}_j$ of $A$:
    \begin{align}
        \ket|b>=\sum_j b_j \ket|\lambda_j>,
        \qquad b_j \coloneqq \braket<\lambda_j|b>.
    \end{align}
    The spectral initialization produces a filtered state
    \begin{align}
        \ket|\tilde{x}>=\sum_j \tilde{f}(\lambda_j) b_j \ket|\lambda_j>,
    \end{align}
    where $\tilde{f}(\cdot)$ approximates $1/\lambda$ on the spectral region retained by the cutoff.

    Let the classical initial iterate be obtained by a decoding map $x_{(0)} = \mathsf{ClassicalDecode}(\ket|\tilde{x}>)$ as discussed in \sec{quantum_accelerated_conjugate_gradient_method}. Consider the action of $A$ on the initializer:
    \begin{align}
        \begin{aligned}
            A\ket|\tilde{x}>
            &= A\sum_j \tilde{f}(\lambda_j) b_j \ket|\lambda_j> \\
            &= \sum_j \lambda_j \tilde{f}(\lambda_j) b_j \ket|\lambda_j>.
        \end{aligned}
    \end{align}
    Hence the residual associated with the initializer is
    \begin{align}
        \begin{aligned}
            r_{(0)}
            &= b - A x_{(0)} \\
            &= \sum_j \Bigl(1-\lambda_j \tilde{f}(\lambda_j)\Bigr)b_j \ket|\lambda_j> \\
            &= \sum_j \alpha_j \ket|\lambda_j>,
        \end{aligned}
    \end{align}
    where
    \begin{align}
        \alpha_j \coloneqq \Bigl(1-\lambda_j \tilde{f}(\lambda_j)\Bigr)b_j.
    \end{align}
    The subsequent CG refinement operates on the Krylov subspace generated by this residual,
    \begin{align}
        \mathcal{K}_m=\text{span}\{r_{(0)},Ar_{(0)},A^2r_{(0)},\dots,A^{m-1}r_{(0)}\},
    \end{align}
    and its convergence rate is governed by an effective condition number $\kappa''$ determined by the spectrum of $A$ restricted to the support of $r_{(0)}$.

    By \thm{time-cg}, the expected runtime of the classical refinement on the HPC model (\defn{HPC}) is
    \begin{align}
        \mathcal{T}_{\text{HPCG}}(T,n,\kappa'',\varepsilon).
    \end{align}
    The two stages are executed sequentially, so the expected total runtime is the sum
    \begin{align}
        \begin{aligned}
            &\mathcal{T}_{\text{QACG}}(T,d,\tau,r,n,\kappa',\kappa'',\varepsilon) = \mathcal{T}_{\text{HHL}}(d,\tau,r,n,\kappa',\varepsilon)+\mathcal{T}_{\text{HPCG}}(T,n,\kappa'',\varepsilon),
        \end{aligned}
    \end{align}
    which proves the claim.
\end{proof}

\section{Performance analysis}
\label{app:performance_analysis}

This section provides additional analysis and supporting evidence for the scaling behavior and resource estimates discussed in the main text. We focus on clarifying the mechanisms underlying the runtime improvement of QACG relative to classical CG and HHL, with particular emphasis on the role of reduction in the condition number, QEC cycle time, and eigenvalue-inversion overheads.

\subsection{Condition number optimization}

\begin{figure}[t]
    \centering
    \includegraphics[width=0.5\linewidth]{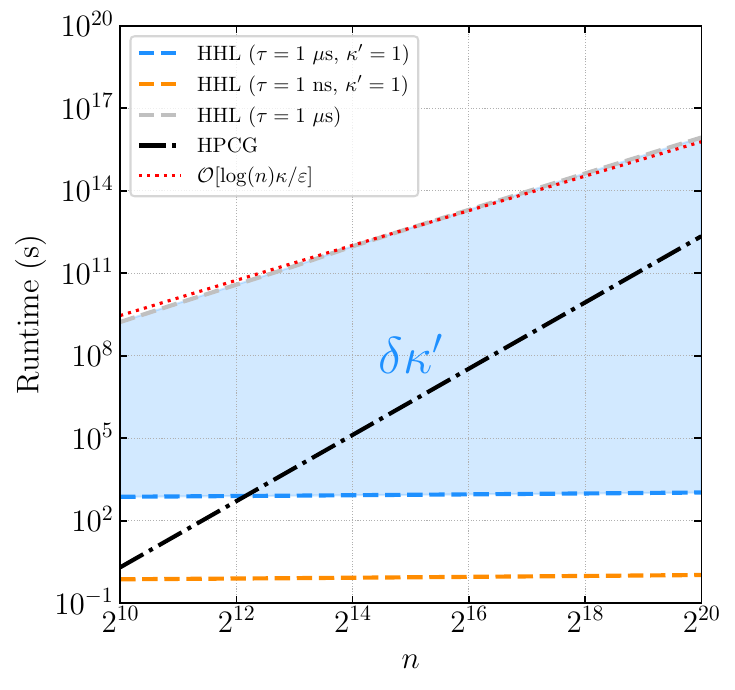}
    \caption{Estimated HHL runtime for the 3D Poisson equation under varying quantum error-correction (QEC) cycle times and quantum condition numbers. 
    Curves show projected runtimes as functions of problem size $n$ for different QEC cycle times $\tau$ and quantum condition numbers $\kappa'$.
    The dotted reference line indicates the scaling $\mathcal{O}[\log(n)\kappa/\varepsilon]$ up to a constant for comparison. 
    Reducing the quantum condition number handled by the quantum subroutine suppresses the explicit $\kappa$-dependence and decreases the slope with respect to $n$. 
    For fixed $\kappa'$, the projected runtime separates parametrically from the $\kappa$-dependent HHL scaling. 
    Decreasing $\tau$ shifts the crossover with the classical benchmark (HPCG) to smaller problem sizes without modifying the asymptotic dependence on $n$.}
    \label{fig:3d_quantum_runtime_theor}
\end{figure}

In \fig{3d_quantum_runtime_theor}, we estimate the expected runtime of HHL for the 3D Poisson equation under different assumptions on the quantum condition number $\kappa'$ and the QEC cycle time $\tau$. The horizontal axis denotes the number of unknowns per spatial dimension $n$, while the vertical axis shows the estimated runtime in seconds. The dash-dotted black curve represents the benchmark-calibrated classical HPCG baseline.

The dashed gray curve corresponds to HHL with $\tau = 1,\mu\text{s}$ and $\kappa'=\kappa$, where the quantum subroutine must resolve the full condition number of the linear system. In this case, the runtime follows a steep dependence on $n$, reflecting the scaling $\kappa=\mathcal{O}(n^2)$ for the 3D Poisson operator (\lem{condition}), together with the reference complexity $\mathcal{O}[\log(n)\kappa/\varepsilon]$ for HHL.

By contrast, the dashed blue and orange curves show HHL with the quantum condition number restricted to $\kappa'=1$, for $\tau=1,\mu\text{s}$ and $\tau=1,\text{ns}$, respectively. In this regime, the explicit $\kappa$-dependence is removed from the quantum stage, leaving only a logarithmic dependence on $n$. As a result, replacing $\kappa=\mathcal{O}(n^2)$ by a constant yields an exponential separation in runtime relative to the classical baseline as $n$ increases. This corrected plot therefore makes explicit that the dominant mechanism for runtime improvement in QACG is the reduction of the condition number handled by the quantum subroutine, rather than the absolute speed of quantum gates alone. Faster QEC cycles act multiplicatively, shifting the crossover to smaller problem sizes without changing the asymptotic slope.

\begin{figure}[t]
    \centering
    \includegraphics[width=0.5\linewidth]{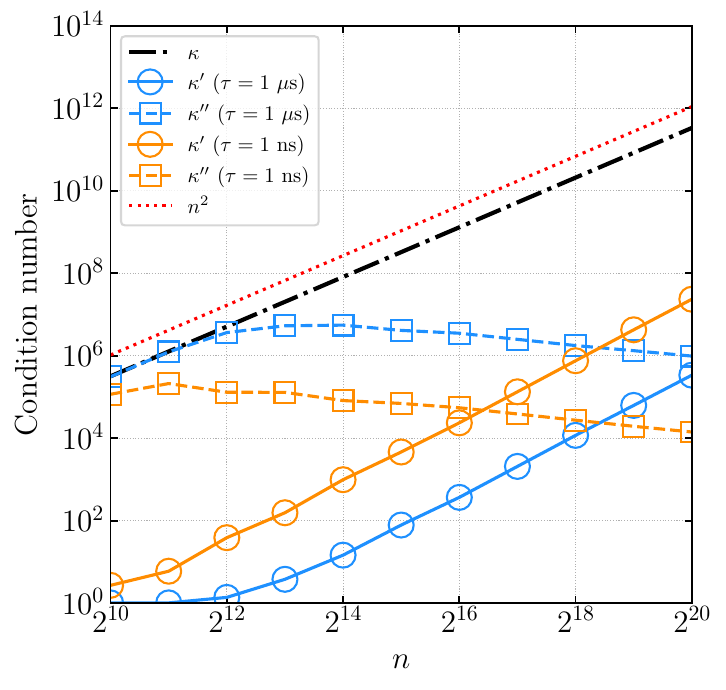}
    \caption{Scaling of the full, quantum, and effective condition numbers for QACG.
    The full condition number $\kappa$, the optimized quantum condition number $\kappa'$, and the effective classical condition number $\kappa''$ are shown as functions of problem size $n$. 
    The dotted line indicates the reference scaling $\kappa = \mathcal{O}(n^2)$ for the 3D Poisson equation.}
    \label{fig:3d_condition_number_scaling}
\end{figure}

\fig{3d_condition_number_scaling} shows how the condition numbers entering QACG evolve with the problem size $n$. The solid black curve indicates the full condition number $\kappa$, which scales as $\mathcal{O}(n^2)$ for the Poisson operator. The colored curves show the runtime-optimized quantum condition number $\kappa'$ and the corresponding effective classical condition number $\kappa''$ for $\tau=1,\mu\text{s}$ and $\tau=1,\text{ns}$.

For small $n$, the optimized $\kappa'$ remains close to unity, indicating that only a very narrow low-energy spectral band is treated on quantum devices. In this regime, QACG behaves similarly to classical CG, whose complexity benefits from the $\sqrt{\kappa}$ dependence of convergence. As $n$ increases, however, $\kappa'$ grows steadily, reflecting the increasing benefit of initializing a broader low-energy subspace on the quantum processor. This growth marks the crossover region in which the logarithmic dependence on $n$ inherent to HHL becomes advantageous relative to the $\sqrt{\kappa}$ scaling of CG.

Correspondingly, the effective classical condition number $\kappa''$ decreases as $\kappa'$ increases, indicating that the slowest classical error modes are progressively removed by the quantum initialization. The separation between $\kappa'$ and $\kappa''$ thus quantifies the redistribution of conditioning between quantum and classical resources, and explains the emergence of a runtime advantage for QACG at large problem sizes.

\subsection{Numerical simulation of spectral initialization}

We also provide a systematic numerical verification of the acceleration mechanism by quantifying how the number of eigenvalues retained in the spectral initialization affects the CG iteration count. Using the same notation as in~\eqref{eq:SpectralInitialization}, let $A$ be symmetric positive definite with spectral decomposition
\begin{equation}
    A = \sum_{i=1}^{n} \lambda_i \ket|\lambda_i>\bra<\lambda_i|, 
    \qquad 0<\lambda_1\le \lambda_2\le \cdots \le \lambda_n,
\end{equation}
and expand the right-hand side as
\begin{equation}
    \ket|b>=\sum_{i=1}^{n}\ket|\lambda_i>\braket<\lambda_i|b>.
\end{equation}
The spectrally truncated inverse is defined by the filtered inverse eigenvalue function
\begin{equation}
    \tilde{f}(\lambda_i) \equiv 
    \begin{cases}
        1/\lambda_i, & i \le N_\lambda,\\
        0, & i > N_\lambda,
    \end{cases}
\end{equation}
which yields the unnormalized initialized state
\begin{align}\label{eq:ClassicalSpectralInitKet}
    \begin{aligned}
        \ket|\tilde{x}>
        &= \tilde{A}^{-1}\ket|b> \\
        &\coloneqq \sum_{i=1}^{n} \tilde{f}(\lambda_i)\ket|\lambda_i>\braket<\lambda_i|b> \\
        &= \sum_{i=1}^{N_\lambda} \frac{1}{\lambda_i}\ket|\lambda_i>\braket<\lambda_i|b>.
    \end{aligned}
\end{align}
Normalizing gives $\tilde{x} = c\ket|\tilde{x}>$ with
\begin{equation}
    c^{-2} =  \sum_{i=1}^{N_\lambda} \frac{|\braket<\lambda_i|b>|^2}{\lambda_i^2}.
\end{equation}

For the classical verification, we use the same spectral truncation to construct the initial CG iterate in the vector representation,
\begin{align}
    \begin{aligned}
        x_{(0)} &= \textsf{ClassicalDecode}\left(\tilde{x}\right) \\
        &\coloneqq \sum_{i=1}^{N_\lambda} \frac{1}{\lambda_i}\ket|\lambda_i>\braket<\lambda_i|b>.
    \end{aligned}
\end{align}
This initialization reconstructs exactly the components of the solution in $\text{span}\{\ket|\lambda_i>\}_{i\le N_\lambda}$, while the orthogonal complement $\text{span}\{\ket|\lambda_i>\}_{i> N_\lambda}$ is resolved by subsequent CG iterations. We then quantify the acceleration by measuring the CG iteration count required to reach a fixed tolerance as a function of $N_\lambda$.

\begin{figure*}[htb]
    \centering
    \includegraphics[width=\linewidth]{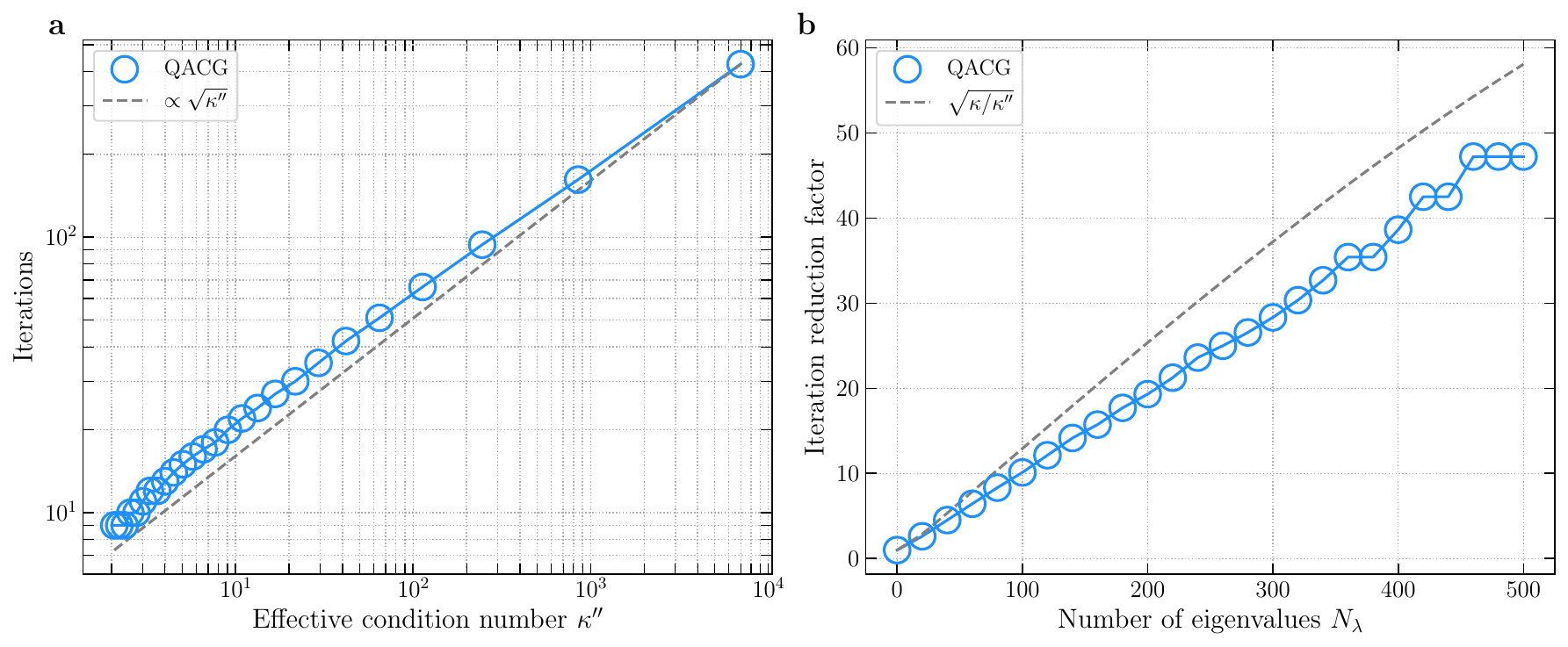}
    \caption{Spectral initialization reduces the effective condition number and accelerates CG in the p--n diode Poisson problem.
    \textbf{a}, CG iteration count as a function of the effective classical condition number $\kappa''$ obtained by retaining $N_\lambda$ low-lying eigenvalues in the initialization. The measured data follow the theoretical scaling $\propto \sqrt{\kappa''}$, confirming that spectral truncation governs the convergence rate. 
    \textbf{b}, Iteration reduction factor $N_{\text{cold}}/N_{\text{warm}}$ as a function of $N_\lambda$, compared with the theoretical prediction $\sqrt{\kappa/\kappa''}$. The measured reduction approaches approximately $80\%$ of the theoretical bound, with deviations attributable to finite stopping tolerance and the discreteness of iteration counts.}
    \label{fig:pn_diode_cg_condition_number}
\end{figure*}

\begin{table}[ht]
    \centering
    \begin{tabular}{ll}
        \hline
        Parameter & Numerical \\
        \hline
        Temperature & $T = 300~\text{K}$ \\
        Doping concentration & $N_\text{A} = N_\text{D} = 10^{16}~\text{cm}^{-3}$ \\
        Device length & $L = 1~\mu\text{m}$ \\
        Number of grid points & $n = 2^{10}$ \\
        Stopping criterion & $\|r\|/\|b\| < 10^{-6}$ \\
        \hline
    \end{tabular}
    \caption{Physical and numerical parameters for the p--n diode Poisson problem.}
    \label{tab:pn_parameters}
\end{table}

We apply this analysis to the Poisson equation for a silicon p--n diode, as described in the main text. The physical and numerical parameters are summarized in \tab{pn_parameters}. The device is evaluated at $T = 300~\text{K}$ with symmetric doping concentrations $N_\text{A} = N_\text{D} = 10^{16}~\text{cm}^{-3}$, device length $L = 1~\mu\text{m}$, and spatial discretization $n = 2^{10}$ grid points. The CG iteration is terminated once the relative residual satisfies $\|r\|/\|b\| < 10^{-6}$.

\fig{pn_diode_cg_condition_number} presents the numerical results obtained by varying the number of retained eigenvalues from $N_\lambda = 0$ (cold start) to $500$. In \figg{pn_diode_cg_condition_number}{a}, the CG iteration count is plotted against the effective condition number $\kappa''$. The results closely follow the theoretical scaling $\propto \sqrt{\kappa''}$, confirming that spectral initialization reduces the effective condition number governing CG convergence. In \figg{pn_diode_cg_condition_number}{b}, the iteration reduction factor $N_{\text{cold}}/N_{\text{warm}}$ is compared with the theoretical prediction $\sqrt{\kappa/\kappa''}$. Deviations become more pronounced for large $N_\lambda$, where the theoretical bound modestly overestimates the achievable reduction.

These deviations may arise from two practical considerations. First, the asymptotic bound $\mathcal{O}(\sqrt{\kappa})$ assumes convergence to arbitrarily high precision. With a fixed tolerance of $\varepsilon = 10^{-6}$, CG terminates once the residual threshold is reached, limiting the observable improvement when $\kappa''$ becomes small. Second, the iteration count is discrete, whereas the theoretical scaling is continuous. 

Despite these effects, the measured reduction approaches approximately $80\%$ of the theoretical maximum, demonstrating substantial practical acceleration from spectral initialization even when ideal asymptotic bounds are not fully realized. These results support the use of spectrally informed initial states within the QACG framework, particularly in low-energy regimes where the runtime crossover can be achieved with early-stage fault-tolerant quantum computers.

\subsection{Parameters used in eigenvalue inversion}

\begin{figure*}[htb]
    \centering
    \includegraphics[width=\linewidth]{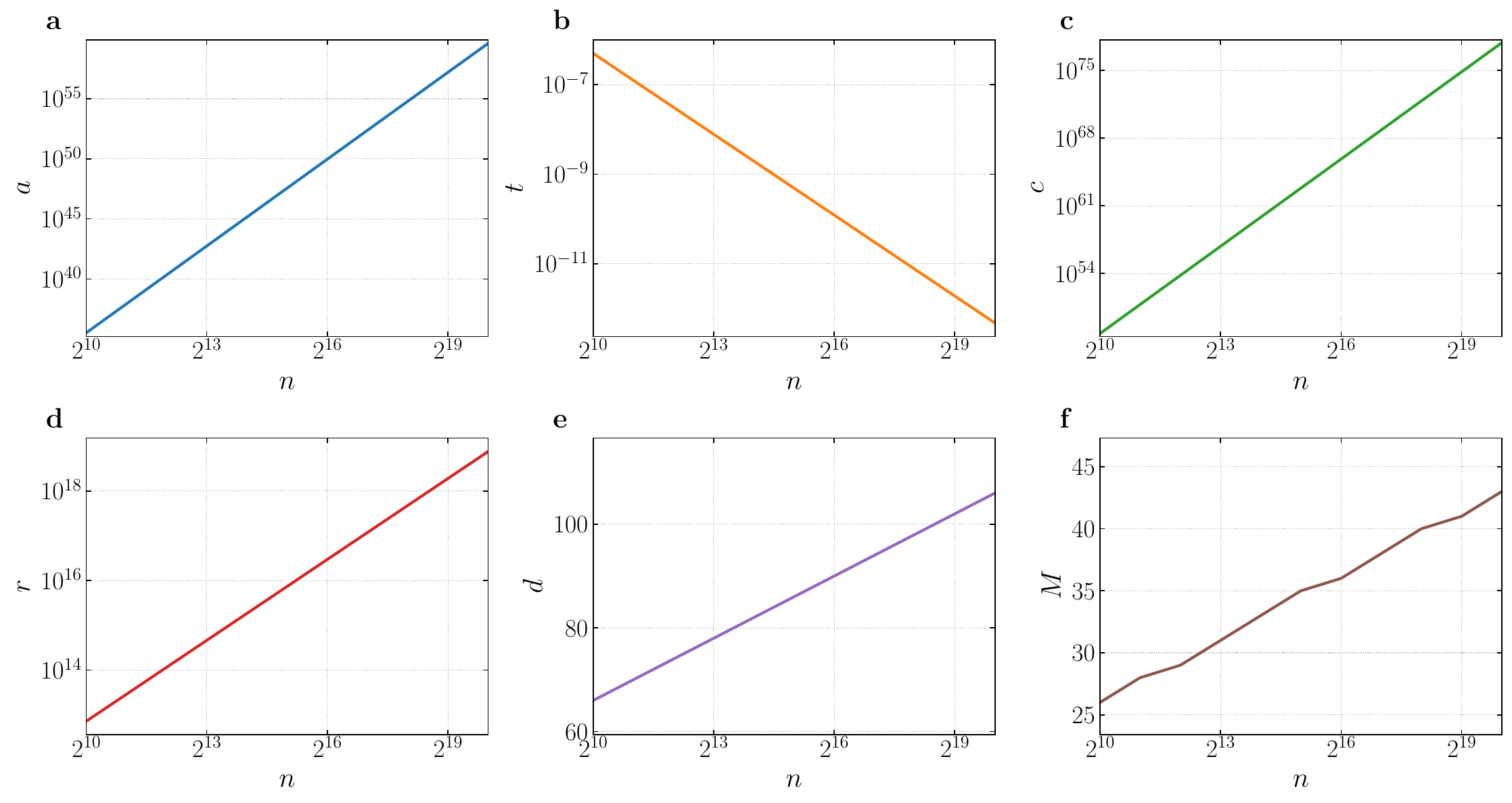}
    \caption{Scaling of parameters in the eigenvalue inversion subroutine as functions of problem size $n$.
    \textbf{a}, Scaling factor $a$. 
    \textbf{b}, Evolution time $t$. 
    \textbf{c}, QPE resolution scale $c$. 
    \textbf{d}, Spectral compression factor $r$. 
    \textbf{e}, Polynomial degree $d$. 
    \textbf{f}, Piecewise segmentation factor $M$.}
    \label{fig:inv_params_scaling}
\end{figure*}

\fig{inv_params_scaling} presents the scaling of the parameters used in the eigenvalue inversion subroutine of HHL. The observed trends are consistent with their definitions in \app{method_of_resource_estimation}. In particular, the scaling of $a$, $c$, and $r$ reflects their explicit dependence on the spectral bounds of the Poisson operator, while the decrease in $t$ with increasing $n$ compensates for the growth of the maximum eigenvalue.

The polynomial degree $d$ grows moderately with $n$, consistent with the increasing precision required to approximate the inverse over a widening spectral range. The parallelization factor $M$ also increases slowly, reflecting the logarithmic dependence on the inverse error tolerance. Importantly, the figure illustrates that the product $d,n_\lambda$ constitutes a dominant contribution to the space complexity of the parallel piecewise-polynomial evaluation used for eigenvalue inversion. This behavior explains why, despite substantial reductions in gate counts, the logical qubit requirement remains significant, as discussed in the resource-estimation analysis.
\section{Source data}
\label{app:source_data}

This section provides the numerical source data for the runtime, logical qubits, and logical gates reported in the main text.

\subsection{Runtime estimation}

\begin{table}[htp]
\centering
\footnotesize
\begin{tabular}{c|c|c|c|c|c|c}
\hline
$n$ & $T_\text{HPCG}$ & $T_\text{HHL}$ & $T_\text{QACG}$ & $\kappa$ & $\kappa'$ & $\kappa''$ \\
\hline
$2^{10}$ & $2.02$ & $1.69 \times 10^{9}$ & $7.50 \times 10^{2}$ & $3.19 \times 10^{5}$ & $1.02$ & $3.11 \times 10^{5}$ \\
$2^{11}$ & $3.24 \times 10^{1}$ & $8.18 \times 10^{9}$ & $8.34 \times 10^{2}$ & $1.27 \times 10^{6}$ & $1.02$ & $1.25 \times 10^{6}$ \\
$2^{12}$ & $5.18 \times 10^{2}$ & $3.85 \times 10^{10}$ & $1.59 \times 10^{3}$ & $5.10 \times 10^{6}$ & $1.39$ & $3.66 \times 10^{6}$ \\
$2^{13}$ & $8.29 \times 10^{3}$ & $1.83 \times 10^{11}$ & $7.99 \times 10^{3}$ & $2.04 \times 10^{7}$ & $3.82$ & $5.34 \times 10^{6}$ \\
$2^{14}$ & $1.33 \times 10^{5}$ & $9.56 \times 10^{11}$ & $5.17 \times 10^{4}$ & $8.16 \times 10^{7}$ & $1.47 \times 10^{1}$ & $5.55 \times 10^{6}$ \\
$2^{15}$ & $2.12 \times 10^{6}$ & $4.46 \times 10^{12}$ & $3.49 \times 10^{5}$ & $3.26 \times 10^{8}$ & $7.90 \times 10^{1}$ & $4.13 \times 10^{6}$ \\
$2^{16}$ & $3.40 \times 10^{7}$ & $2.03 \times 10^{13}$ & $2.49 \times 10^{6}$ & $1.31 \times 10^{9}$ & $3.70 \times 10^{2}$ & $3.53 \times 10^{6}$ \\
$2^{17}$ & $5.43 \times 10^{8}$ & $9.34 \times 10^{13}$ & $1.73 \times 10^{7}$ & $5.22 \times 10^{9}$ & $2.08 \times 10^{3}$ & $2.51 \times 10^{6}$ \\
$2^{18}$ & $8.69 \times 10^{9}$ & $4.27 \times 10^{14}$ & $1.21 \times 10^{8}$ & $2.09 \times 10^{10}$ & $1.18 \times 10^{4}$ & $1.77 \times 10^{6}$ \\
$2^{19}$ & $1.39 \times 10^{11}$ & $1.91 \times 10^{15}$ & $8.35 \times 10^{8}$ & $8.36 \times 10^{10}$ & $6.26 \times 10^{4}$ & $1.33 \times 10^{6}$ \\
$2^{20}$ & $2.23 \times 10^{12}$ & $8.66 \times 10^{15}$ & $5.74 \times 10^{9}$ & $3.34 \times 10^{11}$ & $3.39 \times 10^{5}$ & $9.85 \times 10^{5}$ \\
\hline
\end{tabular}
\caption{Runtime and condition number for the 3D Poisson equation under a $1~\mu\text{s}$ clock-cycle assumption. Columns list the 1D system size $n$, the estimated runtime $T_\text{HPCG}$, $T_\text{HHL}$, and $T_\text{QACG}$ in seconds, the full condition number $\kappa$, the reduced condition number $\kappa'$ assigned to the low-eigenvalue sector treated by HHL after COBYQA optimization of the spectral initialization, and the reduced condition number $\kappa''$ assigned to the complementary high-eigenvalue sector treated by classical CG.}
\label{tab:runtime_comparison_us}
\end{table}

As shown in \tab{runtime_comparison_us}, the $\mu$s clock assumption drives COBYQA toward a conservative spectral split, with a small HHL-resolved low-eigenvalue sector for small and intermediate $n$. This choice keeps $\kappa'$ small, often near unity at the smallest grid sizes, while the remaining classical sector retains a comparatively large $\kappa''$.

\begin{table}[htp]
\centering
\footnotesize
\begin{tabular}{c|c|c|c|c|c|c}
\hline
$n$ & $T_\text{HPCG}$ & $T_\text{HHL}$ & $T_\text{QACG}$ & $\kappa$ & $\kappa'$ & $\kappa''$ \\
\hline
$2^{10}$ & $2.02$ & $1.69 \times 10^{6}$ & $3.18$ & $3.19 \times 10^{5}$ & $2.73$ & $1.17 \times 10^{5}$ \\
$2^{11}$ & $3.24 \times 10^{1}$ & $8.18 \times 10^{6}$ & $2.05 \times 10^{1}$ & $1.27 \times 10^{6}$ & $6.00$ & $2.12 \times 10^{5}$ \\
$2^{12}$ & $5.18 \times 10^{2}$ & $3.85 \times 10^{7}$ & $1.29 \times 10^{2}$ & $5.10 \times 10^{6}$ & $3.90 \times 10^{1}$ & $1.31 \times 10^{5}$ \\
$2^{13}$ & $8.29 \times 10^{3}$ & $1.83 \times 10^{8}$ & $8.97 \times 10^{2}$ & $2.04 \times 10^{7}$ & $1.57 \times 10^{2}$ & $1.30 \times 10^{5}$ \\
$2^{14}$ & $1.33 \times 10^{5}$ & $9.56 \times 10^{8}$ & $6.46 \times 10^{3}$ & $8.16 \times 10^{7}$ & $9.92 \times 10^{2}$ & $8.23 \times 10^{4}$ \\
$2^{15}$ & $2.12 \times 10^{6}$ & $4.46 \times 10^{9}$ & $4.48 \times 10^{4}$ & $3.26 \times 10^{8}$ & $4.66 \times 10^{3}$ & $7.01 \times 10^{4}$ \\
$2^{16}$ & $3.40 \times 10^{7}$ & $2.03 \times 10^{10}$ & $3.08 \times 10^{5}$ & $1.31 \times 10^{9}$ & $2.37 \times 10^{4}$ & $5.50 \times 10^{4}$ \\
$2^{17}$ & $5.43 \times 10^{8}$ & $9.34 \times 10^{10}$ & $2.13 \times 10^{6}$ & $5.22 \times 10^{9}$ & $1.34 \times 10^{5}$ & $3.89 \times 10^{4}$ \\
$2^{18}$ & $8.69 \times 10^{9}$ & $4.27 \times 10^{11}$ & $1.46 \times 10^{7}$ & $2.09 \times 10^{10}$ & $7.59 \times 10^{5}$ & $2.75 \times 10^{4}$ \\
$2^{19}$ & $1.39 \times 10^{11}$ & $1.91 \times 10^{12}$ & $9.97 \times 10^{7}$ & $8.36 \times 10^{10}$ & $4.27 \times 10^{6}$ & $1.96 \times 10^{4}$ \\
$2^{20}$ & $2.23 \times 10^{12}$ & $8.66 \times 10^{12}$ & $6.77 \times 10^{8}$ & $3.34 \times 10^{11}$ & $2.35 \times 10^{7}$ & $1.42 \times 10^{4}$ \\
\hline
\end{tabular}
\caption{Runtime and condition number for the 3D Poisson equation under a $1~\text{ns}$ clock-cycle assumption. Columns list the 1D system size $n$, the estimated runtimes $T_\text{HPCG}$, $T_\text{HHL}$, and $T_\text{QACG}$ in seconds, the full condition number $\kappa$, the reduced condition number $\kappa'$ assigned to the low-eigenvalue sector treated by HHL after COBYQA optimization of the spectral initialization, and the reduced condition number $\kappa''$ assigned to the complementary high-eigenvalue sector treated by classical CG.}
\label{tab:runtime_comparison_ns}
\end{table}

\tab{runtime_comparison_ns} shows that the ns clock assumption shifts the optimized spectral initialization toward a larger quantum-resolved sector. Relative to \tab{runtime_comparison_us}, this produces systematically larger $\kappa'$ and smaller $\kappa''$, consistent with COBYQA allocating more of the low-spectrum burden to HHL when quantum time costs are reduced.

For the runtime data in \tab{runtime_comparison_us} and \tab{runtime_comparison_ns}, the COBYQA optimization of the spectral initialization threshold selects different partitions of the spectrum for the $\tau=1~\mu\text{s}$ and $\tau=1~\text{ns}$ clock-speed assumptions. When the clock cycle is slower ($\tau=1~\mu\text{s}$), the quantum subroutine is more expensive, so COBYQA favors a smaller low-eigenvalue sector for HHL; this choice keeps $\kappa'$ close to unity for small and moderate problem sizes and leaves a larger high-eigenvalue sector to the classical CG stage, reflected in the larger values of $\kappa''$. When the clock cycle is faster ($\tau=1~\text{ns}$), the quantum cost decreases by three orders of magnitude, so COBYQA shifts the splitting point and assigns a broader low-eigenvalue sector to HHL; accordingly, $\kappa'$ becomes larger while $\kappa''$ becomes smaller. This change in the optimized spectral partition explains why QACG remains favorable in both cases, but with different runtime and resource profiles across the two clock-speed regimes.

\subsection{Logical qubits estimation}

\begin{table}[htp]
\centering
\begin{tabular}{c|c|c}
\hline
$n$ & HHL qubits & QACG qubits \\
\hline
$2^{10}$ & $1.2 \times 10^{4}$ & $2.1 \times 10^{3}$ \\
$2^{11}$ & $1.4 \times 10^{4}$ & $2.1 \times 10^{3}$ \\
$2^{12}$ & $1.5 \times 10^{4}$ & $2.2 \times 10^{3}$ \\
$2^{13}$ & $1.7 \times 10^{4}$ & $2.7 \times 10^{3}$ \\
$2^{14}$ & $1.9 \times 10^{4}$ & $3.5 \times 10^{3}$ \\
$2^{15}$ & $2.1 \times 10^{4}$ & $4.6 \times 10^{3}$ \\
$2^{16}$ & $2.3 \times 10^{4}$ & $5.7 \times 10^{3}$ \\
$2^{17}$ & $2.5 \times 10^{4}$ & $7.1 \times 10^{3}$ \\
$2^{18}$ & $2.7 \times 10^{4}$ & $8.6 \times 10^{3}$ \\
$2^{19}$ & $3.0 \times 10^{4}$ & $1.0 \times 10^{4}$ \\
$2^{20}$ & $3.2 \times 10^{4}$ & $1.2 \times 10^{4}$ \\
\hline
\end{tabular}
\caption{Number of logical qubits for HHL and QACG under a $1~\mu\text{s}$ clock-cycle assumption. Columns list the 1D system size $n$, the logical qubits required by HHL, and the logical qubits required by QACG after COBYQA optimization of the spectral initialization.}
\label{tab:qubits_hhl_vs_qacg_us}
\end{table}

In \tab{qubits_hhl_vs_qacg_us}, QACG uses fewer logical qubits than HHL across the full range of grid sizes. Under the slower $\mu$s clock, the optimized spectral split remains relatively modest, so the QACG workspace stays well below the HHL requirement.

\begin{table}[htp]
\centering
\begin{tabular}{c|c|c}
\hline
$n$ & HHL qubits & QACG qubits \\
\hline
$2^{10}$ & $1.2 \times 10^{4}$ & $2.5 \times 10^{3}$ \\
$2^{11}$ & $1.4 \times 10^{4}$ & $3.1 \times 10^{3}$ \\
$2^{12}$ & $1.5 \times 10^{4}$ & $4.1 \times 10^{3}$ \\
$2^{13}$ & $1.7 \times 10^{4}$ & $5.1 \times 10^{3}$ \\
$2^{14}$ & $1.9 \times 10^{4}$ & $6.5 \times 10^{3}$ \\
$2^{15}$ & $2.1 \times 10^{4}$ & $7.8 \times 10^{3}$ \\
$2^{16}$ & $2.3 \times 10^{4}$ & $9.3 \times 10^{3}$ \\
$2^{17}$ & $2.5 \times 10^{4}$ & $1.1 \times 10^{4}$ \\
$2^{18}$ & $2.7 \times 10^{4}$ & $1.3 \times 10^{4}$ \\
$2^{19}$ & $3.0 \times 10^{4}$ & $1.5 \times 10^{4}$ \\
$2^{20}$ & $3.2 \times 10^{4}$ & $1.7 \times 10^{4}$ \\
\hline
\end{tabular}
\caption{Number of logical qubits for HHL and QACG under a $1~\text{ns}$ clock-cycle assumption. Columns list the 1D system size $n$, the logical qubits required by HHL, and the logical qubits required by QACG after COBYQA optimization of the spectral initialization.}
\label{tab:qubits_hhl_vs_qacg_ns}
\end{table}

Compared with \tab{qubits_hhl_vs_qacg_us}, \tab{qubits_hhl_vs_qacg_ns} shows higher QACG qubit counts and smaller reduction factors. This difference follows from the faster ns clock: because quantum operations are reduced, COBYQA selects a broader low-eigenvalue sector for spectral initialization, which increases the quantum memory footprint of QACG.

\subsection{Logical gates estimation}

\begin{table}[htp]
\centering
\footnotesize
\begin{tabular}{c|c|c|c|c|c|c|c|c}
\hline
 & \multicolumn{4}{c|}{HHL gates} & \multicolumn{4}{c}{QACG gates} \\
\cline{2-9}
$n$ & $H$ & $S$ & CNOT & $R_Z$ & $H$ & $S$ & CNOT & $R_Z$ \\
\hline
$2^{10}$ & $1.0 \times 10^{13}$ & $1.8 \times 10^{10}$ & $3.5 \times 10^{13}$ & $3.5 \times 10^{13}$ & $4.3 \times 10^{6}$ & $6.5 \times 10^{4}$ & $1.6 \times 10^{7}$ & $1.6 \times 10^{7}$ \\
$2^{11}$ & $4.9 \times 10^{13}$ & $9.6 \times 10^{10}$ & $1.7 \times 10^{14}$ & $1.7 \times 10^{14}$ & $4.5 \times 10^{6}$ & $9.3 \times 10^{4}$ & $1.6 \times 10^{7}$ & $1.6 \times 10^{7}$ \\
$2^{12}$ & $2.3 \times 10^{14}$ & $4.3 \times 10^{11}$ & $8.0 \times 10^{14}$ & $8.0 \times 10^{14}$ & $4.8 \times 10^{6}$ & $1.0 \times 10^{5}$ & $1.7 \times 10^{7}$ & $1.7 \times 10^{7}$ \\
$2^{13}$ & $1.1 \times 10^{15}$ & $1.9 \times 10^{12}$ & $3.8 \times 10^{15}$ & $3.8 \times 10^{15}$ & $1.3 \times 10^{7}$ & $2.1 \times 10^{5}$ & $4.8 \times 10^{7}$ & $4.8 \times 10^{7}$ \\
$2^{14}$ & $5.7 \times 10^{15}$ & $8.3 \times 10^{12}$ & $2.0 \times 10^{16}$ & $2.0 \times 10^{16}$ & $6.9 \times 10^{7}$ & $8.5 \times 10^{5}$ & $2.5 \times 10^{8}$ & $2.5 \times 10^{8}$ \\
$2^{15}$ & $2.6 \times 10^{16}$ & $3.6 \times 10^{13}$ & $9.3 \times 10^{16}$ & $9.3 \times 10^{16}$ & $5.6 \times 10^{8}$ & $5.0 \times 10^{6}$ & $2.0 \times 10^{9}$ & $2.0 \times 10^{9}$ \\
$2^{16}$ & $1.2 \times 10^{17}$ & $1.6 \times 10^{14}$ & $4.2 \times 10^{17}$ & $4.2 \times 10^{17}$ & $3.9 \times 10^{9}$ & $2.6 \times 10^{7}$ & $1.4 \times 10^{10}$ & $1.4 \times 10^{10}$ \\
$2^{17}$ & $5.5 \times 10^{17}$ & $6.9 \times 10^{14}$ & $1.9 \times 10^{18}$ & $1.9 \times 10^{18}$ & $3.1 \times 10^{10}$ & $1.6 \times 10^{8}$ & $1.1 \times 10^{11}$ & $1.1 \times 10^{11}$ \\
$2^{18}$ & $2.5 \times 10^{18}$ & $3.0 \times 10^{15}$ & $8.9 \times 10^{18}$ & $8.9 \times 10^{18}$ & $2.4 \times 10^{11}$ & $1.0 \times 10^{9}$ & $8.4 \times 10^{11}$ & $8.4 \times 10^{11}$ \\
$2^{19}$ & $1.1 \times 10^{19}$ & $1.3 \times 10^{16}$ & $4.0 \times 10^{19}$ & $4.0 \times 10^{19}$ & $1.6 \times 10^{12}$ & $6.2 \times 10^{9}$ & $5.8 \times 10^{12}$ & $5.8 \times 10^{12}$ \\
$2^{20}$ & $5.1 \times 10^{19}$ & $5.6 \times 10^{16}$ & $1.8 \times 10^{20}$ & $1.8 \times 10^{20}$ & $1.1 \times 10^{13}$ & $3.7 \times 10^{10}$ & $4.0 \times 10^{13}$ & $4.0 \times 10^{13}$ \\
\hline
\end{tabular}
\caption{Number of logical gates for HHL and QACG under a $1~\mu\text{s}$ clock-cycle assumption. Columns list the 1D system size $n$ and the corresponding logical counts of $H$, $S$, CNOT, and $R_Z$ gates for HHL and for QACG using the COBYQA optimizer.}
\label{tab:gates_hhl_vs_qacg_us}
\end{table}

\tab{gates_hhl_vs_qacg_us} shows that QACG requires substantially fewer logical gates than HHL throughout the reported range. Under the $\mu$s clock assumption, the optimized spectral initialization keeps the quantum-resolved sector limited, which suppresses the QACG gate counts.

\begin{table}[htp]
\centering
\footnotesize
\begin{tabular}{c|c|c|c|c|c|c|c|c}
\hline
 & \multicolumn{4}{c|}{HHL gates} & \multicolumn{4}{c}{QACG gates} \\
\cline{2-9}
$n$ & $H$ & $S$ & CNOT & $R_Z$ & $H$ & $S$ & CNOT & $R_Z$ \\
\hline
$2^{10}$ & $1.0 \times 10^{13}$ & $1.8 \times 10^{10}$ & $3.5 \times 10^{13}$ & $3.5 \times 10^{13}$ & $8.3 \times 10^{6}$ & $1.0 \times 10^{5}$ & $3.0 \times 10^{7}$ & $3.0 \times 10^{7}$ \\
$2^{11}$ & $4.9 \times 10^{13}$ & $9.6 \times 10^{10}$ & $1.7 \times 10^{14}$ & $1.7 \times 10^{14}$ & $2.3 \times 10^{7}$ & $2.9 \times 10^{5}$ & $8.2 \times 10^{7}$ & $8.2 \times 10^{7}$ \\
$2^{12}$ & $2.3 \times 10^{14}$ & $4.3 \times 10^{11}$ & $8.0 \times 10^{14}$ & $8.0 \times 10^{14}$ & $2.3 \times 10^{8}$ & $2.1 \times 10^{6}$ & $8.1 \times 10^{8}$ & $8.1 \times 10^{8}$ \\
$2^{13}$ & $1.1 \times 10^{15}$ & $1.9 \times 10^{12}$ & $3.8 \times 10^{15}$ & $3.8 \times 10^{15}$ & $1.2 \times 10^{9}$ & $9.4 \times 10^{6}$ & $4.3 \times 10^{9}$ & $4.3 \times 10^{9}$ \\
$2^{14}$ & $5.7 \times 10^{15}$ & $8.3 \times 10^{12}$ & $2.0 \times 10^{16}$ & $2.0 \times 10^{16}$ & $1.3 \times 10^{10}$ & $6.7 \times 10^{7}$ & $4.6 \times 10^{10}$ & $4.6 \times 10^{10}$ \\
$2^{15}$ & $2.6 \times 10^{16}$ & $3.6 \times 10^{13}$ & $9.3 \times 10^{16}$ & $9.3 \times 10^{16}$ & $7.9 \times 10^{10}$ & $3.5 \times 10^{8}$ & $2.8 \times 10^{11}$ & $2.8 \times 10^{11}$ \\
$2^{16}$ & $1.2 \times 10^{17}$ & $1.6 \times 10^{14}$ & $4.2 \times 10^{17}$ & $4.2 \times 10^{17}$ & $5.1 \times 10^{11}$ & $2.0 \times 10^{9}$ & $1.8 \times 10^{12}$ & $1.8 \times 10^{12}$ \\
$2^{17}$ & $5.5 \times 10^{17}$ & $6.9 \times 10^{14}$ & $1.9 \times 10^{18}$ & $1.9 \times 10^{18}$ & $3.8 \times 10^{12}$ & $1.3 \times 10^{10}$ & $1.4 \times 10^{13}$ & $1.4 \times 10^{13}$ \\
$2^{18}$ & $2.5 \times 10^{18}$ & $3.0 \times 10^{15}$ & $8.9 \times 10^{18}$ & $8.9 \times 10^{18}$ & $2.7 \times 10^{13}$ & $7.9 \times 10^{10}$ & $9.7 \times 10^{13}$ & $9.7 \times 10^{13}$ \\
$2^{19}$ & $1.1 \times 10^{19}$ & $1.3 \times 10^{16}$ & $4.0 \times 10^{19}$ & $4.0 \times 10^{19}$ & $1.9 \times 10^{14}$ & $4.9 \times 10^{11}$ & $6.7 \times 10^{14}$ & $6.7 \times 10^{14}$ \\
$2^{20}$ & $5.1 \times 10^{19}$ & $5.6 \times 10^{16}$ & $1.8 \times 10^{20}$ & $1.8 \times 10^{20}$ & $1.3 \times 10^{15}$ & $3.0 \times 10^{12}$ & $4.5 \times 10^{15}$ & $4.5 \times 10^{15}$ \\
\hline
\end{tabular}
\caption{Number of logical gates for HHL and QACG under a $1~\text{ns}$ clock-cycle assumption. Columns list the 1D system size $n$ and the corresponding logical counts of $H$, $S$, CNOT, and $R_Z$ gates for HHL and for QACG using the COBYQA optimizer.}
\label{tab:gates_hhl_vs_qacg_ns}
\end{table}

Relative to \tab{gates_hhl_vs_qacg_us}, \tab{gates_hhl_vs_qacg_ns} reports larger QACG gate counts because the faster ns clock changes the COBYQA optimum for spectral initialization. With reduced quantum time, the optimizer assigns a larger fraction of the low-eigenvalue spectrum to HHL, which increases the QACG quantum workload even though the overall runtime remains strongly favorable compared with HHL alone.

\bibliographystyle{unsrt}
\bibliography{refs}

\end{document}